\documentclass[11pt,letterpaper,reqno]{amsart}

\title[]{Clebsch Canonization of Lie--Poisson Systems}
\author{Buddhika Jayawardana}
\address{Department of Mathematical Sciences, The University of Texas at Dallas, 800 W Campbell Rd, Richardson, TX 75080-3021, USA}
\email{Buddhika.Jayawardana@utdallas.edu}
\author{Philip J. Morrison}
\address{Department of Physics and Institute for Fusion Studies, The University of Texas at Austin, Austin, TX, 78712, USA}
\email{morrison@physics.utexas.edu}
\author{Tomoki Ohsawa}
\address{Department of Mathematical Sciences, The University of Texas at Dallas, 800 W Campbell Rd, Richardson, TX 75080-3021, USA}
\email{tomoki@utdallas.edu}

\date{\today}

\dedicatory{Dedicated to Professor Anthony Bloch on the occasion of his 65th birthday}



\usepackage{graphicx,amssymb,subfigure,caption,paralist}
\usepackage[margin=1in, marginpar=.5in]{geometry}

\usepackage[numbers,sort&compress]{natbib}
\usepackage[colorlinks=false]{hyperref}

\usepackage[nameinlink]{cleveref}

\usepackage{tikz}
\usepackage{tikz-cd}
\usetikzlibrary{arrows,snakes,backgrounds,patterns,shapes.geometric,calc}
\newcommand\centerofmass{%
    \tikz[radius=0.4em] {%
        \fill (0,0) -- ++(0.4em,0) arc [start angle=0,end angle=90] -- ++(0,-0.8em) arc [start angle=270, end angle=180];%
        \draw (0,0) circle;%
    }%
  }

\theoremstyle{plain}
\newtheorem{theorem}{Theorem}
\newtheorem*{theorem*}{Theorem}

\newtheorem{lemma}[theorem]{Lemma}
\newtheorem{proposition}[theorem]{Proposition}

\theoremstyle{definition}

\theoremstyle{remark}

\def\od#1#2{\frac{d#1}{d#2}}
\def\pd#1#2{\frac{\partial #1}{\partial #2}}

\def\tpd#1#2{\partial #1/\partial #2}

\def\dzero#1#2{\left.\od{}{#1} #2 \right|_{#1=0}}

\def\parentheses#1{{\left(#1\right)}}
\def\brackets#1{{\left[#1\right]}}

\def\anglebrackets#1{{\left<#1\right>}}

\def\tr{\mathop{\mathrm{tr}}\nolimits}

\def\Span{\mathop{\mathrm{span}}\nolimits} 

\def\norm#1{{\left\|#1\right\|}}

\def\DS{\displaystyle}
\def\R{\mathbb{R}}
\def\C{\mathbb{C}}

\def\defeq{\mathrel{\mathop:}=}
\def\eqdef{=\mathrel{\mathop:}}
\def\setdef#1#2{{\left\{ #1 \ |\ #2 \right\}}}
\def\ip#1#2{{\left\langle#1,#2\right\rangle}}

\def\diag{\operatorname{diag}}

\newcommand{\im}{\operatorname{im}}

\def\SO{\mathsf{SO}}
\def\SE{\mathsf{SE}}

\def\gl{\mathfrak{gl}}

\def\orth{\mathfrak{o}}
\def\u{\mathfrak{u}}
\def\se{\mathfrak{se}}
\def\so{\mathfrak{so}}
\def\sp{\mathfrak{sp}}

\def\sym{\mathsf{sym}}

\newenvironment{tbmatrix}{\left[\begin{smallmatrix}}{\end{smallmatrix}\right]}

\def\d{{\bf d}}

\def\PB#1#2{\left\{#1,#2\right\}}

\def\F{\mathbb{F}}
\newcommand\Ad{\operatorname{Ad}}
\newcommand\ad{\operatorname{ad}}

\def\bOmega{\boldsymbol{\Omega}}
\def\bPi{\boldsymbol{\Pi}}
\def\bP{\mathbf{P}}
\def\bGamma{\boldsymbol{\Gamma}}

\begin{document}

\footskip=.6in

\begin{abstract}
  We propose a systematic procedure called the Clebsch canonization for obtaining a canonical Hamiltonian system that is related to a given Lie--Poisson equation via a momentum map.  We describe both coordinate and geometric versions of the procedure, the latter apparently for the first time.  We also find another momentum map so that the pair of momentum maps constitute a dual pair under a certain condition.  The dual pair gives a concrete realization of what is commonly referred to as collectivization of Lie--Poisson systems.  It also implies that solving the canonized system by symplectic Runge--Kutta methods yields so-called collective Lie--Poisson integrators that preserve the coadjoint orbits and hence the Casimirs exactly.  We give a couple of examples, including the Kida vortex and the heavy top on a movable base with controls, which are Lie--Poisson systems on $\mathfrak{so}(2,1)^{*}$ and $(\mathfrak{se}(3) \ltimes \mathbb{R}^{3})^{*}$, respectively.
\end{abstract}

\maketitle

\section{Introduction}
\subsection{The Lie--Poisson Dynamics}
The formalization of mechanics by Lagrange and Hamilton evolved in the 19th century into the  description of dynamical systems  where the  equations of motion are generated by  canonical Poisson brackets,  written in terms of canonical coordinates,  position and momenta,  with  a Hamiltonian function. More modern differential geometric descriptions of Hamiltonian systems occurred well into the 20th century by, e.g.,  \citet{Ma1963} and \citet{Jo1964}, motivating present day symplectic geometry. 

Less well-known is Poisson geometry.
Its origins date back to \citet{li1890} in 1890, but it was brought into modern geometric form with contributions from \citet{So1970} and others including the seminal paper of \citet{We1983}.
Like the canonical Poisson brackets of symplectic geometry, noncanonical Poisson brackets of Poisson geometry are binary operations on smooth phase space functions constituting a Lie algebra realization, but explicit reference to canonical coordinates is removed and degeneracy is allowed.
A manifold with such a Poisson bracket is a generalization of the symplectic manifold called a Poisson manifold.
Noncanonical Poisson brackets, including the present day coordinate-free axioms, were present in the theoretical physics  community in the mid 20th century  in e.g.~the works of \citet{Di1950,Pa1953,Ma1959,Su1964}.

A special kind of noncanonical Poisson bracket, the Lie--Poisson bracket, has explicit linear dependence on the phase space coordinates and is intimately related to a Lie algebra.
Lie--Poisson dynamics---dynamics generated by Lie--Poisson brackets---is ubiquitous as basic equations of physics.
It is this kind of dynamics that is the subject of the present paper. 

An important example of Lie--Poisson dynamics is  given by Euler's equations for rigid body dynamics with a Lie--Poisson bracket based on the Lie algebra of infinitesimal rotations  \cite{Ma1959} (see also \cite{Su1964,SuMu1974}).
This example often serves as inspiration for generalization and exploration of new concepts.
The Lie--Poisson bracket for the full ideal fluid including magnetohydrodynamics was given in \citet{MoGr1980}; see also \citet{ViDo1978}.
This example was followed by the Lie--Poisson formulation of the Maxwell--Vlasov system of equations in  \citet{Mo1980}, with a correction given  in \citet{WeMo1981} and \citet{MaWe1982} and a limitation to the correction pointed out in \citet{Mo1982}, which was followed up more recently in \citet{Mo2013, HeMo2020} and then in \citet{LaSaWe2019}.
Another example from  the mid 1980s is that given in \citet{MaMoWe1984},  where the Lie--Poisson bracket for general moment closures of  the kinetic hierarchy were given.
There is now a large literature with very many subsequent publications that can be found, e.g., in  \citet{Mo1998,Mo2006} and \citet{ArKh1998}. 

Given the ubiquity of the Lie--Poisson form, it is natural to inquire about its origin.
One thread extends back to the quasi-coordinate description of \citet{Po1901b} (see also \citet{Ha1904}), where Euler's equations for rigid body dynamics, and its Lagrangian counterpart---the Euler--Poincar\'e equation---was first formulated on a general Lie algebra.
This idea was applied to fluid dynamics in \citet{Ar1966,Ar1969} where Euler's equations for the incompressible fluid are seen to be the Euler--Poincar\'e equation on the Lie algebra of a diffeomorphism group, putting the work of \citet{La1788} into modern language. (See \citet{MoAnPe2020} for commentary.)
Although these works did not explicitly give the Lie--Poisson bracket, the equations of motion for a reduced dynamics were obtained. 

The main geometrical idea behind the Lie--Poisson brackets for the dynamics of rigid body, fluids, and plasmas is now understood as a process of reduction from canonical to noncanonical Hamiltonian form as follows  (see, e.g., \citet[Chapter~13]{MaRa1999}): The configuration space of the systems is a Lie group $\mathsf{G}$, and the basic equation of the system is a canonical  Hamiltonian system defined on the cotangent bundle $T^{*}\mathsf{G}$.
However, the Hamiltonian has $\mathsf{G}$-symmetry, and thus one may reduce the system to the dual $\mathfrak{g}^{*}$ of the Lie algebra $\mathfrak{g}$ of $\mathsf{G}$.
The resulting equation on $\mathfrak{g}^{*}$  has  Lie--Poisson form.

There are also examples where the system is defined on a Lie group $\mathsf{G}$, but the symmetry of the system is broken.  A well-known example is the heavy top, where the symmetry is broken by the gravity;  another is the  compressible fluid, where density plays a  role similar to gravity for the heavy top case.
In either case, it is known that one can still recover the full symmetry by extending the configuration space to a semidirect product $\mathsf{G} \ltimes V$ using a $\mathsf{G}$-representation on a vector space $V$; see, e.g., \citet{MaRaWe1984a,MaRaWe1984b,HoMaRa1998a}.
As a result, one again obtains a Lie--Poisson system on the dual of the semidirect product Lie algebra $\mathfrak{g} \ltimes V$, which is a special case of Lie--Poisson brackets based on Lie algebra extensions \cite{ThMo2000} that occur in a variety of physical systems including magnetohydrodynamics (see \citet{MaMo1984}).

\subsection{Collectivization}
Another class of Lie--Poisson systems arises as a result of  so-called \textit{collectivization} in the sense of \citet{GuSt1980} (see also \citet{HoMa1983} and \citet[Section~28]{GuSt1990}).
Given a Poisson manifold $P$ and an equivariant momentum map $\mathbf{M}\colon P \to \mathfrak{g}^{*}$ associated with an action of a Lie group $\mathsf{G}$ on $P$, one can show that $\mathbf{M}$ is a Poisson map with respect to the Poisson bracket on $P$ and the Lie--Poisson bracket on $\mathfrak{g}^{*}$; see, e.g., \citet[Theorem~12.4.1]{MaRa1999}.
This implies the following:
Given that Hamiltonian $H\colon P \to \R$ is \textit{collective} in the sense that there exists $h \colon \mathfrak{g}^{*} \to \R$ such that $H = h \circ \mathbf{M}$, the flow $\Phi_{t}$ of the Hamiltonian vector field on $P$ defined by $H$ and the flow $\phi_{t}$ of the Lie--Poisson dynamics on $\mathfrak{g}^{*}$ defined by $h$ are related by $\mathbf{M}$ as $\mathbf{M} \circ \Phi_{t} = \phi_{t} \circ \mathbf{M}$.

The term ``collective'' comes from the motivating examples of \citet{GuSt1980,GuSt1990}) such as the liquid drop model in nuclear physics, where one seeks a set of equations that describe aggregate motions of a number of particles ``as if it were a rigid body or liquid drop''; the idea behind this dates back to \citet{Ri1860} (see also \citet{Ro1988} and \citet{MoLeBi2009}).

\subsection{Clebsch Canonization and Collectivization}
What we refer to as Clebsch canonization or just ``\textit{canonization}'' for short  is the opposite of the collectivization described above: One first has a Lie--Poisson equation on $\mathfrak{g}^{*}$, and then constructs a cotangent bundle $T^{*}Q$ and an equivariant momentum map $\mathbf{M}\colon T^{*}Q \to \mathfrak{g}^{*}$ so that solutions of the new \textit{canonical} Hamiltonian dynamics on $T^{*}Q$ can be mapped by $\mathbf{M}$ to those of the Lie--Poisson dynamics on $\mathfrak{g}^{*}$.

This theoretical concept is motivated by the early use of potentials for describing the velocity field of fluid mechanics:  long before the introduction of the vector potential for representing a magnetic field, researchers considered various potential representations of velocity fields, the most famous  of which  is due to Clebsch \cite{Cl1857,Cl1859}.
The connection between the Lie--Poisson brackets for fluid dynamics and the canonical Hamiltonian description in terms of  the Clebsch representation was first  given in \citet{Mo1981,Mo1982,MoGr1982}, while two-dimensional vortex dynamics was considered later in \citet{MaWe1983}.
See also \citet{Oh2019e} for the Clebsch representation of the heavy top dynamics.
A general theory for Lie--Poisson brackets, motivated by \cite{Mo1981} was given in  \citet{Mo1998} and the present work places this in the geometric setting described above.

\subsection{Lie--Poisson Integrators}
Compared to symplectic integrators for canonical Hamiltonian systems (see, e.g., \citet{HaLuWa2006} and \citet{LeRe2004}), integrators for Lie--Poisson equations seem to be studied less extensively.
Some earlier works include \citet{GeMa1988} and \citet{ChSc1991}, and are based on generating functions.
\citet{EnFa2001} used Lie group methods by exploiting the property that Lie--Poisson dynamics evolves on coadjoint orbits on $\mathfrak{g}^{*}$.
More recently, \citet{MaRo2010} developed a variational integrator for the Lie--Poisson equation by discretizing the corresponding variational principle.
See also \citet{Ma-LPI} for a more recent survey of Lie--Poisson integrators.

Our work gives a concrete realization of the general theory of \textit{collective integrators} developed by \citet{McMoVe2014}; see also \citet{McMoVe2015,McMoVe2016}.
The main advantage of collective integrators is that one can construct Lie--Poisson integrators that preserve the coadjoint orbits out of existing symplectic integrators.
On the other hand, the main disadvantage is that it is not always clear how one can find a suitable cotangent bundle $T^{*}Q$ and momentum map $\mathbf{M}$.

It is important that the symplectic integrator ``descends''~\cite{McMoVe2014} to a Lie--Poisson integrator that preserves the coadjoint orbits.
One can show that this is the case with the symplectic Runge--Kutta method if, for example, one can find another momentum map $\mathbf{J}$ on $T^{*}Q$ so that the pair of momentum maps $\mathbf{M}$ and $\mathbf{J}$ constitute a dual pair, as discussed in \cite[Theorem~7]{McMoVe2014}.
Existing constructions (see, e.g., \citet{McMoVe2014,McMoVe2015,McMoVe2016}) of such momentum maps $\mathbf{M}$ and $\mathbf{J}$ are rather ad-hoc, and thus are limited to Lie--Poisson equations on relatively simple spaces such as $\orth(p, q, \F)$, $\sp(2k, F)$, $\gl(n, \F)$, $\u(p, q)$ with $\F = \R, \C, \mathbb{H}$, and some semi-direct products.

\subsection{Main Result and Outline}
We propose a systematic canonization that potentially works for a wider class of Lie--Poisson equations by constructing a momentum map $\mathbf{M}\colon T^{*}\mathfrak{g} \to \mathfrak{g}^{*}$; hence the Lie--Poisson equation on $\mathfrak{g}^{*}$ is ``canonized'' to a canonical Hamiltonian system $T^{*}\mathfrak{g} \cong T^{*}\R^{n}$ with $n \defeq \dim\mathfrak{g}$.
We first show how this works in coordinate calculations in \Cref{sec:canonization}.

In \Cref{sec:geometry}, we give a geometric interpretation of this setting.
We also find a Lie subalgebra $\mathfrak{h}$ of $\sp(2n,\R)$ that characterizes the intrinsic symmetry of the canonized Hamiltonian system (or the \textit{canonized system} for short).
Its action on $T^{*}\mathfrak{g}$ gives rise to another momentum map $\mathbf{J}\colon T^{*}\mathfrak{g} \to \mathfrak{h}^{*}$ that becomes invariants of the canonized system.
We then prove in \Cref{thm:dual_pair} that the momentum maps $\mathbf{M}$ and $\mathbf{J}$ constitute a dual pair (in the sense of \citet{We1983}) under a certain condition.

\Cref{sec:properties} addresses the invariants of the canonized system.
For any (real) Lie algebra $\mathfrak{g}$, the momentum map $\mathbf{J}$ has at least two components including an invariant associated with the Killing form on $\mathfrak{g}$.
Additionally, if $\mathfrak{g}$ is semisimple, then there is another invariant associated with the Killing form.
Furthermore, we show that if the Lie--Poisson bracket on $\mathfrak{g}^{*}$ possesses a Casimir then there is a corresponding Noether-type invariant (momentum map) in the canonized system as well.

In \Cref{sec:integrators}, we first briefly review the idea of the collective integrators, and then show some numerical results.
Assuming the dual pair from \Cref{thm:dual_pair}, symplectic Runge--Kutta methods applied to our canonized system yields Lie--Poisson integrators that preserve the coadjoint orbits and hence the Casimirs exactly.
We demonstrate it using a couple of examples: the Kida vortex~\cite{Ki1981} (see also \citet{MeMoFl1997}) and the heavy top on a movable base with a stabilizing control~\cite{CoOh-EPwithBSym1}.

\section{Clebsch Canonization}
\label{sec:canonization}
\subsection{Lie--Poisson Bracket}
Let $\mathfrak{g}$ be an $n$-dimensional Lie algebra, and $\{ E_{i} \}_{i=1}^{n}$ be a basis for it with the structure constants $\{ c_{ij}^{k} \}_{1\le i,j,k\le n}$, i.e., $[E_{i}, E_{j}] = c_{ij}^{k} E_{k}$; note that we use Einstein's summation convention throughout the paper.
We may then define the dual basis $\{ E_{*}^{i} \}_{i=1}^{n}$ for $\mathfrak{g}^{*}$ by setting $\ip{ E_{*}^{i} }{ E_{j} } = \delta^{i}_{j}$ under the standard dual pairing $\ip{\,\cdot\,}{\,\cdot\,}\colon \mathfrak{g}^{*} \times \mathfrak{g} \to \R$.

For any smooth $f \colon \mathfrak{g}^{*} \to \R$, we define the derivative $Df(\mu) \in \mathfrak{g}$ evaluated at $\mu \in \mathfrak{g}^{*}$ so that, for any $\delta\mu \in \mathfrak{g}^{*}$,
\begin{equation*}
  \ip{ \delta\mu }{ Df(\mu) }
  = \dzero{s}{ f(\mu + s\delta\mu) }.
\end{equation*}
This results in the coordinate expression
\begin{equation*}
  Df(\mu) = \pd{f}{\mu_{i}}(\mu)\, E_{i}.
\end{equation*}
Then one defines the $(+)$-Lie--Poisson bracket (see \Cref{ssec:(-)LPB} for the $(-)$-Lie--Poisson bracket) on $\mathfrak{g}^{*}$ as follows:
For any $f, g \colon \mathfrak{g}^{*} \to \R$,
\begin{equation}
  \label{eq:LPB+}
  \PB{f}{g}_{+}(\mu) \defeq \ip{ \mu }{ \brackets{ Df(\mu), Dg(\mu) } }
  = \mu_{k} c^{k}_{ij} \pd{f}{\mu_{i}} \pd{g}{\mu_{j}}.
\end{equation}

The Lie--Poisson equation for a Hamiltonian $h \colon \mathfrak{g}^{*} \to \R$ is the Hamiltonian system defined using the above Poisson bracket, i.e.,
\begin{subequations}
  \label{eq:LP+}
  \begin{equation}
    \dot{\mu}_{i} = \PB{\mu_{i}}{h}_{+}
    = \mu_{k} c^{k}_{ij} \pd{h}{\mu_{j}},
  \end{equation}
  or equivalently,
  \begin{equation}
    \dot{\mu} = -\ad_{Dh(\mu)}^{*}\mu.
  \end{equation}
\end{subequations}

\subsection{Clebsch Canonization in Coordinates}
\label{ssec:CcanonizationC}
The main idea of the Clebsch canonization (see \citet{Mo1981,Mo1998}) is the following: Given an $n$-dimensional Lie--Poisson bracket~\eqref{eq:LP+}, we would like to find a corresponding $2n$-dimensional canonical Hamiltonian system.
In other words, we would like to find a relationship between the Poisson bracket~\eqref{eq:LPB+} and the canonical Poisson bracket of the form
\begin{equation}
  \label{eq:canPB}
  \PB{F}{G} = \pd{F}{q^{i}} \pd{G}{p_{i}} - \pd{G}{q^{i}} \pd{F}{p_{i}},
\end{equation}
where $F, G \colon T^{*}\R^{n} \to \R$.

Suppose that $\mu = \mu_{i} E_{*}^{i} \in \mathfrak{g}^{*}$ and $(q,p) \in T^{*}\R^{n}$ are related as follows:
\begin{equation}
  \label{eq:mu-qp}
  \mu_{i} = c^{k}_{ij} q^{j} p_{k}.
\end{equation}
For any smooth $f, g \colon \mathfrak{g}^{*} \to \R$, we may define $F, G \colon T^{*}\R^{n} \to \R$ by setting $F(q,p) \defeq f(\mu)$ where $\mu$ and $(q,p)$ are related as above; similarly for $G$ as well.
Then, by the chain rule, we have
\begin{equation*}
  \pd{F}{q^{i}} = \pd{f}{\mu_{j}} \pd{\mu_{j}}{q^{i}} = \pd{f}{\mu_{j}} c^{k}_{ji} p_{k},
  \qquad
  \pd{F}{p_{i}} = \pd{f}{\mu_{j}} \pd{\mu_{j}}{p_{i}} = \pd{f}{\mu_{j}} c^{i}_{jk} q^{k}.
\end{equation*}
As a result,
\begin{align*}
  \PB{F}{G}
  &= \pd{F}{q^{i}} \pd{G}{p_{i}} - \pd{G}{q^{i}} \pd{F}{p_{i}} \\
  &= q^{l} p_{m} \parentheses{ c^{i}_{kl} c^{m}_{ji} - c^{i}_{jl} c^{m}_{ki} } \pd{f}{\mu_{j}} \pd{g}{\mu_{k}} \\
  &= q^{l} p_{m} \parentheses{ - c^{i}_{kl} c^{m}_{ij} - c^{i}_{lj} c^{m}_{ik} } \pd{f}{\mu_{j}} \pd{g}{\mu_{k}} \\
  &= q^{l} p_{m} c^{i}_{jk} c^{m}_{il} \pd{f}{\mu_{j}} \pd{g}{\mu_{k}} \\
  &= \mu_{i} c^{i}_{jk} \pd{f}{\mu_{j}} \pd{g}{\mu_{k}} \\
  &= \PB{f}{g}_{+}(\mu),
\end{align*}
where the fourth equality follows from the Jacobi identity for the structure constants.

Therefore, given any Lie--Poisson bracket in terms of $\mu$, one can obtain an canonized canonical bracket via \eqref{eq:mu-qp}.
The Hamiltonian of the canonized system will have the form $H(q,p) = h(\mu)$ with $\mu$ given as in \eqref{eq:mu-qp}.
If the resulting equations of the canonical system are solved for $t \mapsto (q(t),p(t))$, then $t \mapsto \mu(t)$ constructed according to \eqref{eq:mu-qp} solves the Lie--Poisson equation~\eqref{eq:LP+}.

\section{Geometry of Clebsch Canonization}
\label{sec:geometry}
This section gives a geometric interpretation of the canonization presented in \Cref{ssec:CcanonizationC}.
Particularly, we show that the map~\eqref{eq:mu-qp} is the momentum map associated with a natural $\mathfrak{g}$-action on the cotangent bundle $T^{*}\mathfrak{g}$.

\subsection{Left $\mathfrak{g}$-action on $T^{*}\mathfrak{g}$}
Let $T^{*}\mathfrak{g} = \mathfrak{g} \times \mathfrak{g}^{*}$ be the cotangent bundle of $\mathfrak{g}$ and define
\begin{equation*}
  \mathfrak{g} \times T^{*}\mathfrak{g} \to T^{*}\mathfrak{g};
  \qquad
  (\xi,(q,p)) \mapsto (\ad_{\xi}q, -\ad_{\xi}^{*}p) \eqdef \xi_{T^{*}\mathfrak{g}}(q,p).
\end{equation*}
In coordinates, we can write it as follows:
\begin{align*}
  \xi_{T^{*}\mathfrak{g}}(q,p)
  &= \xi^{i} c_{ij}^{k} q^{j} \pd{}{q^{k}} - \xi^{i} c_{ij}^{k} p_{k} \pd{}{p_{j}} \\
  &= \xi^{i} c_{ij}^{k} \parentheses{ q^{j} \pd{}{q^{k}} - p_{k} \pd{}{p_{j}} }.
\end{align*}

Let us show that it is a left Lie algebra action, i.e., for any $\xi, \eta \in \mathfrak{g}$,
\begin{equation*}
  [\xi,\eta]_{T^{*}\mathfrak{g}} = -[\xi_{T^{*}\mathfrak{g}}, \eta_{T^{*}\mathfrak{g}}],
\end{equation*}
where the bracket on the left-hand side is the commutator in $\mathfrak{g}$ whereas the one on the right is the Jacobi--Lie bracket of vector fields on $T^{*}\mathfrak{g}$.
In fact, in the coordinate representation with respect to the standard basis $\{\tpd{}{q^{i}}, \tpd{}{p_{i}}\}_{i=1}^{n}$, we have,
\begin{align*}
  D\eta_{T^{*}\mathfrak{g}} \cdot \xi_{T^{*}\mathfrak{g}}
  &= \parentheses{
    \pd{}{q} (\ad_{\eta}q) \cdot \ad_{\xi}q,\,
    -\pd{}{p} (\ad_{\eta}^{*}p) \cdot (-\ad_{\xi}^{*}p)
  } \\
  &= \parentheses{ \ad_{\eta} \circ \ad_{\xi} q,\, \ad_{\eta}^{*} \circ \ad_{\xi}^{*} p } \\
  &= \parentheses{ [\eta,[\xi,q]],\, \ad_{\eta}^{*} \circ \ad_{\xi}^{*} p },
\end{align*}
where the second line follows because $q \mapsto \ad_{\eta}q$ and $p \mapsto \ad_{\eta}^{*}p$ are linear.
Therefore, we obtain
\begin{align*}
  [\xi_{T^{*}\mathfrak{g}}, \eta_{T^{*}\mathfrak{g}}]
  &= D\eta_{T^{*}\mathfrak{g}} \cdot \xi_{T^{*}\mathfrak{g}} - D\xi_{T^{*}\mathfrak{g}} \cdot \eta_{T^{*}\mathfrak{g}} \\
  &= \parentheses{
    [\eta,[\xi,q]] - [\xi,[\eta,q]],\,
    \ad_{\eta}^{*} \circ \ad_{\xi}^{*} p - \ad_{\xi}^{*} \circ \ad_{\eta}^{*} p
  } \\
  &= \parentheses{ [q,[\xi,\eta]],\, \ad_{[\xi,\eta]}^{*} p } \\
  &= \parentheses{ -\ad_{[\xi,\eta]} q,\, \ad_{[\xi,\eta]}^{*} p } \\
  &= -[\xi,\eta]_{T^{*}\mathfrak{g}},
\end{align*}
where we used the Jacobi identity of the commutator on $\mathfrak{g}$ and the following dual version of it: for any $\xi, \eta \in \mathfrak{g}$,
\begin{equation}
  \label{eq:dual_Jacobi}
  \ad_{\eta}^{*} \circ \ad_{\xi}^{*} - \ad_{\xi}^{*} \circ \ad_{\eta}^{*} = \ad_{[\xi,\eta]}^{*}.
\end{equation}

If $\mathfrak{g}$ is the Lie algebra of a Lie group $\mathsf{G}$, then we may first consider the left $\mathsf{G}$-action on $T^{*}\mathfrak{g}$ as follows:
\begin{equation*}
  \Phi\colon \mathsf{G} \times T^{*}\mathfrak{g} \to T^{*}\mathfrak{g};
  \qquad
  (g,(q,p)) \mapsto (\Ad_{g}q, \Ad_{g^{-1}}^{*}p) \defeq \Phi_{g}(q,p).
\end{equation*}
Clearly this is the cotangent lift of the adjoint action of $\mathsf{G}$ on $\mathfrak{g}$.
Then its infinitesimal generator gives the above Lie algebra action:
\begin{equation*}
  \left.\od{}{s} \Phi_{\exp(s\xi)}(q,p) \right|_{s=0} = (\ad_{\xi}q, -\ad_{\xi}^{*}p) = \xi_{T^{*}\mathfrak{g}}(q,p).
\end{equation*}

\subsection{Momentum Map $\mathbf{M}^{+}$}
Let us find the momentum map associated with the above Lie algebra action.
For any $\xi \in \mathfrak{g}$, define $M_{\xi}\colon T^{*}\mathfrak{g} \to \R$ by setting
\begin{equation*}
  X_{M_{\xi}} = \xi_{T^{*}\mathfrak{g}},
\end{equation*}
where $X_{M_{\xi}}$ is the Hamiltonian vector field for $M_{\xi}$ with respect to the canonical symplectic form on $T^{*}\mathfrak{g}$, i.e.,
\begin{equation*}
  X_{M_{\xi}} = \pd{M_{\xi}}{p_{i}} \pd{}{q^{i}} - \pd{M_{\xi}}{q^{j}} \pd{}{p_{j}}.
\end{equation*}
It is a straightforward calculation to find
\begin{equation*}
  M_{\xi}(q,p) = \ip{ p }{ \ad_{\xi} q } = -\ip{ \ad_{q}^{*}p }{ \xi }.
\end{equation*}
The momentum map $\mathbf{M}^{+}\colon T^{*}\mathfrak{g} \to \mathfrak{g}^{*}$ is then defined so that
\begin{equation*}
  \ip{ \mathbf{M}^{+}(q,p) }{ \xi } = M_{\xi}(q,p),
\end{equation*}
which yields
\begin{equation}
  \label{eq:M_+}
  \mathbf{M}^{+}(q,p) = -\ad_{q}^{*}p.
\end{equation}
We can obtain a coordinate expression for $\mathbf{M}^{+}$ using the dual basis $\{ E_{*}^{i} \}_{i=1}^{n}$ for $\mathfrak{g}^{*}$ as follows:
\begin{equation}
  \label{eq:M_+-coordinates}
 \mathbf{M}^{+}(q,p) = -q^{j} c_{ji}^{k} p_{k}\,E_{*}^{i} = c_{ij}^{k} q^{j} p_{k}\,E_{*}^{i},
\end{equation}
which is nothing but \eqref{eq:mu-qp} obtained earlier.

The above momentum map is infinitesimally equivariant:
For any $\eta \in \mathfrak{g}$ and any $(q,p) \in T^{*}\mathfrak{g}$,
\begin{align*}
  T_{(q,p)}\mathbf{M}^{+} \cdot \eta_{T^{*}\mathfrak{g}}(q,p)
  &= -\ad_{[\eta,q]}^{*}p + \ad_{q}^{*} \ad_{\eta}^{*} p \\
  &= \ad_{\eta}^{*} \ad_{q}^{*} p \\
  &= -\ad_{\eta}^{*} \mathbf{M}^{+}(q,p),
\end{align*}
where we again used the dual version~\eqref{eq:dual_Jacobi} of the Jacobi identity.
The infinitesimal equivariance implies (see, e.g., \citet[Theorem 12.4.1]{MaRa1999}) that $\mathbf{M}^{+}$ is a Poisson map with respect to the canonical Poisson bracket \eqref{eq:canPB} on $T^{*}\mathfrak{g} \cong T^{*}\R^{n}$ and the $(+)$-Lie--Poisson bracket \eqref{eq:LPB+} on $\mathfrak{g}^{*}$, i.e., for any smooth $f, g \colon \mathfrak{g}^{*} \to \R$,
\begin{equation}
  \label{eq:Poisson+}
  \PB{ f \circ \mathbf{M}^{+} }{ g \circ \mathbf{M}^{+} } = \PB{f}{g}_{+} \circ \mathbf{M}^{+}.
\end{equation}

\subsection{Right Action and $(-)$-Lie--Poisson bracket}
\label{ssec:(-)LPB}
In order to find a Poisson map $\mathbf{M}^{-}$ with respect to the $(-)$-Lie--Poisson bracket
\begin{equation}
  \label{eq:LPB-}
  \PB{f}{g}_{-}(\mu) = -\ip{ \mu }{ \brackets{ Df(\mu), Dg(\mu) } }
  = -\mu_{k} c^{k}_{ij} \pd{f}{\mu_{i}} \pd{g}{\mu_{j}}
\end{equation}
on $\mathfrak{g}^{*}$, one starts instead with the following right Lie algebra action:
\begin{equation*}
  \mathfrak{g} \times T^{*}\mathfrak{g} \to T^{*}\mathfrak{g};
  \qquad
  (\xi,(q,p)) \mapsto (-\ad_{\xi}q, \ad_{\xi}^{*}p) \eqdef \xi_{T^{*}\mathfrak{g}}(q,p),
\end{equation*}
which satisfies $[\xi,\eta]_{T^{*}\mathfrak{g}} = [\xi_{T^{*}\mathfrak{g}}, \eta_{T^{*}\mathfrak{g}}]$.

If $\mathfrak{g}$ is the Lie algebra of a Lie group $\mathsf{G}$, then we may consider the following right $\mathsf{G}$-action on $T^{*}\mathfrak{g}$:
\begin{equation*}
  \Phi\colon \mathsf{G} \times T^{*}\mathfrak{g} \to T^{*}\mathfrak{g};
  \qquad
  (g,(q,p)) \mapsto (\Ad_{g^{-1}}q, \Ad_{g}^{*}p) \defeq \Phi_{g}(q,p).
\end{equation*}
Then we have
\begin{equation*}
  \left.\od{}{t} \Phi_{\exp(t\xi)}(q,p) \right|_{t=0} = (-\ad_{\xi}q, \ad_{\xi}^{*}p) = \xi_{T^{*}\mathfrak{g}}(q,p).
\end{equation*}

The associated momentum map is
\begin{equation}
  \label{eq:M_-}
  \mathbf{M}^{-}(q,p) = \ad_{q}^{*}p = -c_{ij}^{k} q^{j} p_{k}\,E_{*}^{i},
\end{equation}
and satisfies, for any $f, g \colon \mathfrak{g}^{*} \to \R$,
\begin{equation}
  \label{eq:Poisson-}
  \PB{ f \circ \mathbf{M}^{-} }{ g \circ \mathbf{M}^{-} } = \PB{f}{g}_{-} \circ \mathbf{M}^{-}.
\end{equation}

\subsection{Clebsch Canonization}
Summarizing the above arguments, we have the following special class of \textit{symplectic realization} or \textit{Clebsch variables} (see, e.g., \citet{MaWe1983}):
\begin{theorem}[Clebsch canonization of Lie--Poisson equations]
  \label{thm:main}
  Given a smooth function $h\colon \mathfrak{g}^{*} \to \R$, define $H \colon T^{*}\mathfrak{g} \cong T^{*}\R^{n} \to \R$ as
  \begin{equation*}
    H(q,p) \defeq h \circ \mathbf{M}^{\pm}(q,p) = h\parentheses{ \mp\ad_{q}^{*}p },
  \end{equation*}
  using $\mathbf{M}^{\pm}\colon T^{*}\mathfrak{g} \to \mathfrak{g}^{*}$ defined in \eqref{eq:M_+} or \eqref{eq:M_-}, respectively.
  Let $t \mapsto (q(t), p(t))$ be a solution to the canonical Hamiltonian system (referred to as the canonized system)
  \begin{equation}
    \label{eq:canonized}
    \dot{q} = \pd{H}{p},
    \qquad
    \dot{p} = -\pd{H}{q}
  \end{equation}
  on $T^{*}\mathfrak{g} \cong T^{*}\R^{n}$.
  Then $t \mapsto \mu(t) \defeq \mathbf{M}^{\pm}(q(t),p(t))$ gives a solution to the Lie--Poisson equation
  \begin{equation*}
    \dot{\mu} = \mp\ad_{Dh(\mu)}^{*}\mu,
  \end{equation*}
  defined in terms of the $(\pm)$-Lie--Poisson bracket, \eqref{eq:LPB+} or \eqref{eq:LPB-}, respectively.
\end{theorem}
\begin{proof}
   It easily follows from the property that $\mathbf{M}^{\pm}$ is Poisson (see \eqref{eq:Poisson+} and \eqref{eq:Poisson-}) with respect to the canonical Poisson bracket on $T^{*}\mathfrak{g}$ and the $(\pm)$-Lie--Poisson bracket on $\mathfrak{g}^{*}$, respectively.
\end{proof}

\subsection{Dual Pair}
Moreover, there exists a dual pair of momentum maps in the sense of \citet{We1983} associated with the above canonization.
In order to define the other momentum map, let us first write $\mathbf{M}^{+}$ from \eqref{eq:M_+-coordinates} (the case with $\mathbf{M}^{-}$ is virtually the same) using its components as follows:
Writing $z = (q,p)$ for short,
\begin{equation*}
  \mathbf{M}^{+}(z) = M_{i}(z)\,E_{*}^{i},
  \qquad
  M_{i}(z) \defeq q^{T} \mathcal{C}_{i} p,
\end{equation*}
where $\{ \mathcal{C}_{i} \}_{i=1}^{n}$ are the $n \times n$ matrices defined in terms of the structure constants as
\begin{equation*}
  (\mathcal{C}_{i})_{jk} \defeq c_{ij}^{k}.
\end{equation*}
However, we may also write $M_{i}(z)$ as a bilinear form on $T^{*}\mathfrak{g}$ as follows:
\begin{equation}
  \label{eq:M-components}
  M_{i}(z) = \frac{1}{2} z^{T} \mathcal{M}_{i} z,
\end{equation}
where $\{ \mathcal{M}_{i} \}_{i=1}^{n}$ are the $2n \times 2n$ symmetric matrices defined as
\begin{equation*}
  \mathcal{M}_{i} \defeq
  \begin{bmatrix}
    0 & \mathcal{C}_{i} \\
    \mathcal{C}_{i}^{T} & 0
  \end{bmatrix}.
\end{equation*}

Consider the symplectic algebra
\begin{equation*}
  \sp(2n,\R) \defeq \setdef{ \tilde{\xi} \in \R^{2n\times2n} }{ \tilde{\xi}^{T} \mathbb{J} + \mathbb{J} \tilde{\xi} = 0 }
  \quad\text{where}\quad
  \mathbb{J} \defeq
  \begin{bmatrix}
    0 & I_{n} \\
    -I_{n}& 0
  \end{bmatrix}.
\end{equation*}
It is well known that $\sp(2n,\R)$ can be identified with the set $\sym(2n,\R)$ of $2n \times 2n$ real symmetric matrices equipped with the Lie bracket
\begin{equation*}
  [\xi,\eta]_{\mathbb{J}} \defeq \xi \mathbb{J} \eta - \eta \mathbb{J} \xi
\end{equation*}
via the following map:
\begin{equation*}
  \sym(2n,\R) \to \sp(2n,\R);
  \qquad
  \xi \mapsto \tilde{\xi} \defeq \mathbb{J} \xi.
\end{equation*}
Now, let us define a subalgebra $\mathfrak{h}$ of $\sp(2n,\R)$ as follows:
\begin{equation}
  \label{eq:mathfrak_h}
  \begin{split}
    \mathfrak{h}
    &\defeq \setdef{ \tilde{\sigma} \in \sp(2n,\R) }{ \tilde{\sigma}^{T} \mathcal{M}_{i} + \mathcal{M}_{i} \tilde{\sigma} = 0\ \forall i \in \{1, \dots, n\} } \\
    &\cong \setdef{ \sigma \in \sym(2n,\R) }{ [\sigma, \mathcal{M}_{i}]_{\mathbb{J}} = 0\ \forall i \in \{1, \dots, n\} } \\
    &= \setdef{ \sigma \in \sym(2n,\R) }{ (\sigma \mathbb{J} \mathcal{M}_{i})^{T} = -\sigma \mathbb{J} \mathcal{M}_{i}\ \forall i \in \{1, \dots, n\} }.
  \end{split}
\end{equation}
More concretely, we may write
\begin{equation*}
  \sigma =
  \begin{bmatrix}
    \sigma_{11} & \sigma_{12} \\
    \sigma_{12}^{T} & \sigma_{22}
  \end{bmatrix}
  \in \sym(2n,\R),
\end{equation*}
and see the following characterization of $\mathfrak{h}$:
\begin{equation}
  \label{eq:mathfrak_h-characterization}
  \sigma \in \mathfrak{h}
  \iff
  \left\{
  \begin{array}{l}
    \mathcal{C}_{i} \sigma_{11} = -\sigma_{11} \mathcal{C}_{i}^{T}, \smallskip\\
    \mathcal{C}_{i} \sigma_{12} = \sigma_{12} \mathcal{C}_{i}, \smallskip\\
    \sigma_{22} \mathcal{C}_{i} = -\mathcal{C}_{i}^{T} \sigma_{22}
  \end{array}
  \right.
  \forall i\in\{1, \dots, n\}.
\end{equation}

One can show that $\mathfrak{h}$ is non-trivial for any Lie algebra $\mathfrak{g}$:
\begin{proposition}
  \label{prop:nontriviality_of_h}
  Let $\kappa$ be the $n\times n$ symmetric matrix defining the Killing form on $\mathfrak{g}$, i.e.,
  \begin{equation*}
    \kappa(x, y) \defeq \tr(\ad_{x} \circ \ad_{y}) = \kappa_{ij} x^{i} y^{j},
  \end{equation*}
  or, in terms of the structure constants, $\kappa_{ij} \defeq c_{ik}^{l} c_{jl}^{k}$.
  Then, the following elements of $\sym(2n,\R)$ are contained in $\mathfrak{h}$:
  \begin{equation*}
    \sigma_{0} \defeq
    \begin{bmatrix}
      0 & I_{n} \\
      I_{n} & 0
    \end{bmatrix},
    \qquad
    \varkappa \defeq
    \begin{bmatrix}
      \kappa & 0 \\
      0 & 0
    \end{bmatrix}.
  \end{equation*}
  Furthermore, if $\mathfrak{g}$ is semisimple, then
  \begin{equation*}
    \varkappa^{*} \defeq
    \begin{bmatrix}
      0 & 0 \\
      0 & \kappa^{-1}
    \end{bmatrix}
  \end{equation*}
  is also contained in $\mathfrak{h}$ as well.
\end{proposition}
\begin{proof}
  See \Cref{sec:proof-nontriviality_of_h}.
\end{proof}

Using the above subalgebra $\mathfrak{h}$, we can construct the following dual pair in the sense of \citet{We1983} with some additional assumptions:
\begin{theorem}[Dual pair associated with Clebsch canonization]
  \label{thm:dual_pair}
  Let $\mathfrak{h}$ be the subalgebra of $\sp(2n,\R) \cong \sym(2n,\R)$ defined in \eqref{eq:mathfrak_h}, and consider the $\mathfrak{h}$-action on $T^{*}\mathfrak{g}$ defined by
  \begin{equation*}
    \mathfrak{h} \to \mathfrak{X}(T^{*}\mathfrak{g});
    \qquad
    \sigma \mapsto \sigma_{T^{*}\mathfrak{g}}(z) \defeq \tilde{\sigma} z = \mathbb{J} \sigma z,
  \end{equation*}
  and let $\mathbf{J}\colon T^{*}\mathfrak{g} \to \mathfrak{h}^{*}$ be its associated momentum map.
  Then:
  \begin{enumerate}[(i)]
  \item $J_{\sigma}(z) \defeq \ip{ \mathbf{J}(z) }{ \sigma } = \frac{1}{2}z^{T} \sigma z$ for any $z \in T^{*}\mathfrak{g}$ and any $\sigma \in \mathfrak{h}$.
    \smallskip
  \item There exists an open subset $U \subset T^{*}\mathfrak{g}$ such that $\mathbf{M}^{+}$ (or $\mathbf{M}^{-}$) and $\mathbf{J}$ are both submersions.
    \smallskip
  \item If $U$ is non-empty and $\dim\mathfrak{h} = \dim\mathfrak{g}$, then
    \begin{equation*}
      \begin{tikzcd}
        \mathfrak{h}^{*} & U \arrow[swap]{l}{\mathbf{J}} \arrow{r}{\mathbf{M}^{+}} & \mathfrak{g}^{*}
      \end{tikzcd}
    \end{equation*}
    is a dual pair with respect to the standard symplectic form $\Omega$ on $T^{*}\mathfrak{g}$ (restricted to $U$), i.e.,
    \begin{equation*}
      \parentheses{ \ker T_{z}\mathbf{M}^{+} }^{\Omega} = \ker T_{z}\mathbf{J}
      \quad
      \forall z \in U,
    \end{equation*}
    and similarly with $\mathbf{M}^{-}$ in place of $\mathbf{M}^{+}$, where $(\,\cdot\,)^{\Omega}$ stands for the symplectically orthogonal complement.
  \end{enumerate}
\end{theorem}
\begin{proof}
  See \Cref{sec:proof-dual_pair}.
\end{proof}

\section{Properties of Canonization}
\label{sec:properties}
\subsection{Subalgebra $\mathfrak{h}$ and Momentum Map $\mathbf{J}$}
The dual pair constructed in \Cref{thm:dual_pair} implies that the momentum map $\mathbf{J}$ is an invariant of the canonized system~\eqref{eq:canonized}.
For example, for $\sigma_{0}, \varkappa, \varkappa^{*} \in \mathfrak{h}$ from \Cref{prop:nontriviality_of_h}, the corresponding invariants are, writing $z = (q,p)$ for short,
\begin{equation}
  \label{eq:J_0}
  J_{0}(z) \defeq J_{\sigma_{0}}(z) = p \cdot q,
\end{equation}
and
\begin{equation}
  \label{eq:Killing_inv}
  J_{\varkappa}(z) = q^{T}\kappa q = \kappa(q,q),
  \qquad
  J_{\varkappa^{*}}(z) = p^{T} \kappa^{-1} p = \kappa^{-1}(p,p),
\end{equation}
where we abused the notation by using $\kappa$ and $\kappa^{-1}$ for both the bilinear forms and the associated matrices.

Note that, depending on the Lie algebra $\mathfrak{g}$, the subalgebra $\mathfrak{h} \subset \sym(2n,\R)$ may be larger than $\Span\{ \kappa_{0}, \varkappa \}$ or $\Span\{ \kappa_{0}, \varkappa, \varkappa^{*} \}$, and so there may be more invariants, as we shall see in the example presented in \Cref{ssec:htmb}, where $\dim\mathfrak{h} = 9$.

\subsection{Casimirs and Momentum Maps}
If the Lie--Poisson bracket possesses a Casimir, then there must be a corresponding invariant for the canonical Hamiltonian system~\eqref{eq:canonized}.
We would like to show that the invariant is indeed a Noether invariant (momentum map) of the canonized system:
\begin{proposition}
  \label{prop:Casimir-momentum_map}
  Suppose that $f\colon \mathfrak{g}^{*} \to \R$ is a Casimir of the Lie--Poisson bracket~\eqref{eq:LPB+} or \eqref{eq:LPB-}, and define $F\colon T^{*}\mathfrak{g} \to \R$ by setting $F \defeq f \circ \mathbf{M}^{\pm}$, i.e.,
  \begin{equation*}
    F(q,p) \defeq f\parentheses{ \mp\ad_{q}^{*}p }.
  \end{equation*}
  Let us also define
  \begin{equation*}
    \gamma \colon T^{*}\mathfrak{g} \to \mathfrak{g};
    \qquad
    \gamma(q,p)  \defeq Df(\mathbf{M}^{\pm}(q,p)) = Df\parentheses{ \mp\ad_{q}^{*}p },
  \end{equation*}
  and consider the following $\R$ (Lie algebra) action
  \begin{equation*}
    \R \times T^{*}\mathfrak{g} \to \mathfrak{X}(T^{*}\mathfrak{g});
    \qquad
    (s, (q,p)) \mapsto (\ad_{s\gamma(q,p)}q, -\ad_{s\gamma(q,p)}^{*}p) \defeq s_{T^{*}\mathfrak{g}}(q,p),
  \end{equation*}
  where $\mathfrak{X}(T^{*}\mathfrak{g})$ stands for the space of vector fields on $T^{*}\mathfrak{g}$.
  Then the momentum map corresponding to the action is $F$.
  Furthermore, the Hamiltonian $H$ is infinitesimally invariant under the action, and thus $F$ is an invariant of the canonized system~\eqref{eq:canonized}.
\end{proposition}
\begin{proof}
  Notice first that
  \begin{equation*}
    X_{F}(q,p) = \parentheses{ \pd{F}{p},\, -\pd{F}{q} }
    = \pm\parentheses{ \ad_{\gamma(q,p)}q,\, -\ad_{\gamma(q,p)}^{*}p },
  \end{equation*}
  and so, for any $s \in \R$,
  \begin{equation*}
    X_{s F} = s_{T^{*}\mathfrak{g}}.
  \end{equation*}
  This shows that $F$ is the momentum map corresponding to the above symmetry.

  Let us show that $\mathbf{M}^{\pm}$ is infinitesimally invariant under the above $\R$-action:
  First note that, since $f$ is a Casimir, its derivative $Df$ satisfies $\ad_{Df(\mu)}^{*}\mu = 0$ for any $\mu \in \mathfrak{g}^{*}$; see, e.g., \citet[Corollary~14.4.3]{MaRa1999}.
  Therefore, setting $\mu = -\ad_{q}^{*}p$ in particular, we have
  \begin{equation*}
    -\ad_{\gamma(q,p)}^{*} \ad_{q}^{*} p
    = \ad_{Df\parentheses{ -\ad_{q}^{*}p }}^{*} (-\ad_{q}^{*} p) = 0
  \end{equation*}
  for any $(q,p) \in T^{*}\mathfrak{g}$.
  Then, for any $s \in \R$, the directional derivative of $\mathbf{M}^{\pm}$ along the vector field $s_{T^{*}\mathfrak{g}}$ yields
  \begin{align*}
    s_{T^{*}\mathfrak{g}}[\mathbf{M}^{\pm}](q,p)
    &= \mp s \parentheses{ \ad_{[\gamma(q,p), q]}^{*}p - \ad_{q}^{*}\ad_{\gamma(q,p)}^{*}p } \\
    &= \mp s \parentheses{
      \ad_{\gamma(q,p)}^{*} \ad_{q}^{*} p
      - \ad_{q}^{*} \ad_{\gamma(q,p)}^{*} p
      + \ad_{q}^{*}\ad_{\gamma(q,p)}^{*} p
      } \\
    &= \mp s\parentheses{ \ad_{\gamma(q,p)}^{*} \ad_{q}^{*} p }\\
    &= 0,
  \end{align*}
  where we used the dual version~\eqref{eq:dual_Jacobi} of the Jacobi identity in the second equality.
  
  This implies that the Hamiltonian $H \defeq h \circ \mathbf{M}^{\pm}$ is infinitesimally invariant under the $\R$-action as well.
  That $F$ is an invariant of \eqref{eq:canonized} follows easily from either that $\mathbf{M}^{\pm}$ is Poisson or Noether's Theorem (see, e.g., \citet[Theorem~11.4.1]{MaRa1999}).
\end{proof}

\section{Collective Integrators via Clebsch Canonization}
\label{sec:integrators}
\subsection{Collective Lie--Poisson Integrators via Clebsch Canonization}
Let $\Psi_{\Delta t}\colon T^{*}\mathfrak{g} \to T^{*}\mathfrak{g}$ be an integrator with time step $\Delta t$ for the canonized system~\eqref{eq:canonized}.
In order for the resulting Lie--Poisson integrator to be \textit{collective} in the sense of \citet{McMoVe2014,McMoVe2015}, the method $\Psi_{\Delta t}$ must ``descend'' to a Lie--Poisson integrator $\psi_{\Delta t}$ on $\mathfrak{g}^{*}$ such that $\psi_{\Delta t} \circ \mathbf{M}^{\pm} = \mathbf{M}^{\pm} \circ \Psi_{\Delta t}$ and also that preserves coadjoint orbits in $\mathfrak{g}^{*}$ (and hence its Casimirs) exactly.

According to \citet[Theorem~7]{McMoVe2014}, one of the possible realizations of collective integrators is to have a dual pair of momentum maps $\mathbf{M}^{\pm}$ and $\mathbf{J}$ where $\mathbf{J}$ is quadratic, and use any symplectic Runge--Kutta method for $\Psi_{\Delta t}$.
Since \Cref{thm:dual_pair} gives the desired form of dual pair, the symplectic Runge--Kutta methods applied to our setting gives a collective integrators on $\mathbf{M}^{\pm}(U) \subset \mathfrak{g}^{*}$.

We use the Gauss--Legendre methods---a family of implicit Runge--Kutta methods based on the points of Gauss--Legendre quadrature---as the symplectic integrator $\Psi_{\Delta t}$ for the canonized system~\eqref{eq:canonized}.
The order of a Gauss--Legendre method is $2s$ if it is based on $s$ points \cite[Theorem 5.2]{HaNoWa1993}; the simplest is of order 2 and is the Implicit Midpoint Method.
In this paper, we will use the $4^{\rm th}$ order Gauss--Legendre method; see, e.g., \citet[Table~6.4 on p.~154]{LeRe2004}.

\subsection{Example~1: Kida Vortex}
The Kida vortex~\cite{Ki1981} is an elliptical vortex patch of constant vorticity in a two-dimensional flow.
The equations of motion obtained by \citeauthor{Ki1981} describe the time evolution of the semi-major axis $a$ and semi-minor axis $b$ and of the angle $\phi$ of orientation of the ellipse in a steady shear background flow:
\begin{equation*}
  \dot{a} = \frac{\epsilon}{2} a \sin(2\phi),
  \qquad
  \dot{b} = -\frac{\epsilon}{2} b \sin(2\phi),
  \qquad
  \dot{\phi} = \frac{a b}{(a + b)^{2}} + \frac{\omega}{2} + \frac{\epsilon}{2} \frac{a^{2} + b^{2}}{a^{2} - b^{2}} \cos(2\phi),
\end{equation*}
where $\epsilon > 0$ is the constant rate of strain of the background shear flow.
Defining the aspect ratio $\lambda \defeq b/a$, the equations reduce to
\begin{equation}
  \label{eq:Kida}
  \dot{\lambda} = -\epsilon\lambda \sin(2\phi),
  \qquad
  \dot{\phi} = \frac{\lambda}{(1 + \lambda)^{2}} + \frac{\omega}{2} + \frac{\epsilon}{2} \frac{1 + \lambda^{2}}{1 - \lambda^{2}} \cos(2\phi).
\end{equation}
It is then not difficult to see that the above system of equations is Hamiltonian~\cite{MeFlSe1989,MeZaSt1986}.

Furthermore, \citet{MeMoFl1997} showed that \eqref{eq:Kida} follows from a Lie--Poisson equation on $\so(2,1)^{*}$ obtained by projecting the Lie--Poisson structure for the 2D incompressible Euler equation onto quadratic moments of the vorticity.
Specifically, let $\so(2,1)$ be the Lie algebra of the Lie group
\begin{equation*}
  \SO(2,1) \defeq \setdef{ R \in \R^{3\times3} }{ R^{T} K R = K }
  \text{ with }
  K \defeq
  \begin{bmatrix}
    1 & 0 & 0 \\
    0 & 1 & 0 \\
    0 & 0 & -1
  \end{bmatrix}.
\end{equation*}
A basis for $\so(2,1)$ is given by
$\{ E_{1} = \begin{tbmatrix}
  0 & 0 & 0 \\
  0 & 0 & 1 \\
  0 & 1 & 0
\end{tbmatrix},
E_{2} = \begin{tbmatrix}
  0 & 0 & 1 \\
  0 & 0 & 0 \\
  1 & 0 & 0
\end{tbmatrix},
E_{3} = \begin{tbmatrix}
  0 & -1 & 0 \\
  1 & 0 & 0 \\
  0 & 0 & 0
\end{tbmatrix} \}$,
for which the structure constants $\{c^{k}_{ij}\}_{1\le i,j,k\le 3}$ satisfy, for any $\mu \in \so(2,1)^{*} \cong \R^{3}$,
\begin{equation*}
  \mu_{k} c^{k}_{ij} =
  \begin{bmatrix}
    0 & \mu_{3} & \mu_{2} \\
    -\mu_{3} & 0 & -\mu_{1} \\
    -\mu_{2} & \mu_{1} & 0
  \end{bmatrix}.
\end{equation*}
This is the (class A) type VIII Lie algebra of the Bianchi classification~\cite{ElMa1969,YoToMo2017}.
The Casimir of the corresponding Lie--Poisson bracket~\eqref{eq:LPB+} is then
\begin{equation}
  \label{eq:Casimir-Kida}
  f_{1}(\mu) \defeq \mu_{1}^{2} + \mu_{2}^{2} - \mu_{3}^{2},
\end{equation}
which is essentially the area of the ellipse.

The Killing form in this case is
\begin{equation*}
  \kappa(x,y) = 2(x_{1} y_{1} + x_{2} y_{2} - x_{3} y_{3}).
\end{equation*}
It is clearly non-degenerate, and thus there are two additional invariants (see \eqref{eq:Killing_inv}):
\begin{equation}
  \label{eq:Killing_inv-Kida}
  J_{1}(q,p) \defeq \frac{1}{2}\kappa(q,q) = q_{1}^{2} + q_{2}^{2} - q_{3}^{2},
  \qquad
  J_{2}(q,p) \defeq \frac{1}{2}\kappa^{-1}(p,p) = p_{1}^{2} + p_{2}^{2} - p_{3}^{2}.
\end{equation}
It is also easy to show that $\mathfrak{h} = \Span\{ \sigma_{0}, \varkappa, \varkappa^{*} \}$ using \eqref{eq:mathfrak_h-characterization} (see also \Cref{prop:nontriviality_of_h}); hence $J_{0}$ from \eqref{eq:J_0} along with these two invariants are the components of the momentum map $\mathbf{J}$.

The variables $(\mu_{1}, \mu_{2}, \mu_{3})$ are related to the original variables $(\lambda, \phi)$ as follows:
\begin{equation}
  \label{eq:mu-lambda_phi}
  \mu_{2} = \frac{\pi}{16}\parentheses{ \lambda - \frac{1}{\lambda} } \cos(2\phi),
  \qquad
  \mu_{3} = -\frac{\pi}{16}\parentheses{ \lambda + \frac{1}{\lambda} } \cos(2\phi),
  \qquad
  \mu_{1}^{2} + \mu_{2}^{2} - \mu_{3}^{2} = -\frac{\pi^{2}}{64}.
\end{equation}
With the Hamiltonian (the ``excess energy'' of the elliptical vortex patch~\cite{MeMoFl1997}) $h\colon \so(2,1)^{*} \cong \R^{3} \to \R$ defined as
\begin{equation}
  \label{eq:h-Kida}  
  h(\mu) \defeq \epsilon \mu_{2} + \omega \mu_{3} - \frac{\pi}{8} \ln\parentheses{ \frac{\pi}{8} - \mu_{3} },
\end{equation}
the Lie--Poisson equation $\dot{\mu} = -\ad_{Dh(\mu)}^{*}\mu$ from \eqref{eq:LP+} yields
\begin{equation}
  \label{eq:LP-Kida}
  \dot{\mu}_{1} = \omega \mu_{2} + \epsilon \mu_{3} + \frac{\pi \mu_{2}}{\pi - 8\mu_{3}},
  \qquad
  \dot{\mu}_{2} = -\mu_{1}\parentheses{ \omega +  \frac{\pi}{\pi - 8\mu_{3}} },
  \qquad
  \dot{\mu}_{3} = \epsilon \mu_{1}.
\end{equation}
One can then show that, using \eqref{eq:mu-lambda_phi}, the above Lie--Poisson equation gives rise to the original equation \eqref{eq:Kida} of \citeauthor{Ki1981}.

The map~\eqref{eq:M_+} yields (lowering the indices for $q$ for simplicity),
\begin{equation}
  \label{eq:M-Kida}
  \mathbf{M}^{+}(q,p) = (q_{2} p_{3} + q_{3} p_{2},\, -q_{3} p_{1} - q_{1} p_{3},\, -q_{1} p_{2} + q_{2} p_{1}).
\end{equation}
Following the proof (in \Cref{sec:proof-dual_pair}) of \Cref{thm:dual_pair}, we can show that there exists an open set $U$ that is dense in $T^{*}\mathfrak{g}$ on which $\mathbf{M}^{+}$ and $\mathbf{J}$ are submersions.
We also saw above that $\dim\mathfrak{h} = 3 = \dim\mathfrak{g}$.
Hence we have a dual pair as described in \Cref{thm:dual_pair}.

We then have the Hamiltonian
\begin{align*}
  H(q,p) &\defeq h(\mathbf{M}^{+}(q,p)) \\
  &= -\epsilon (q_{3} p_{1} + q_{1} p_{3}) - \omega (q_{1} p_{2} - q_{2} p_{1}) - \frac{\pi}{8} \ln\parentheses{ \frac{\pi}{8} + q_{1} p_{2} - q_{2} p_{1} }.
\end{align*}
The canonized system~\eqref{eq:canonized} is therefore
\begin{equation}
  \label{eq:canonized-Kida}
  \begin{array}{lll}
    \DS \dot{q}_{1} = \omega q_{2} - \epsilon q_{3} + \frac{\pi}{8} \frac{q_{2}}{q_{1} p_{2} - q_{2} p_{1} + \pi/8},
    & \DS \dot{q}_{2} = -\omega q_{1} - \frac{\pi}{8} \frac{q_{1}}{q_{1} p_{2} - q_{2} p_{1} + \pi/8},
    & \DS \dot{q}_{3} = -\epsilon q_{1},
    \bigskip\\
    \DS \dot{p}_{1} = \omega p_{2} + \epsilon p_{3} + \frac{\pi}{8} \frac{p_{2}}{q_{1} p_{2} - q_{2} p_{1} + \pi/8},
    & \DS \dot{p}_{2} = -\omega p_{1} - \frac{\pi}{8} \frac{p_{1}}{q_{1} p_{2} - q_{2} p_{1} + \pi/8},
    & \DS \dot{p}_{3} = \epsilon p_{1}.
  \end{array}
\end{equation}

\Cref{fig:mu-Kida} shows numerical results with parameters $\epsilon = 1/2$ and $\omega = -1$ with initial condition determined by $\mu_{1}(0) = 1$, $f_{1}(\mu(0)) = -1/4$ and $h(\mu(0)) = 1$; this is a case from \citet[Fig.~2]{MeMoFl1997}.
It shows the time evolution of the solution to \eqref{eq:LP-Kida} computed by the collective integrator as well as the trajectory of the solution in $\so(2,1)^{*}$ plotted with the level sets of the Hamiltonian $h$ and the Casimir $f_{1}$; see \eqref{eq:h-Kida} and \eqref{eq:Casimir-Kida}.
We used the $4^{\rm th}$ order Gauss--Legendre method to solve the canonized system~\eqref{eq:canonized-Kida} with the initial condition $(q(0),p(0))$ obtained by solving $\mathbf{M}^{+}(q(0),p(0)) = \mu(0)$; we additionally imposed $q(0) = (1,0,0)$ and $p_{1}(0) = 0$ to obtain the unique solution.

\begin{figure}[htbp]
  \centering
  \subfigure[Time evolution]{
    \includegraphics[width=.475\linewidth]{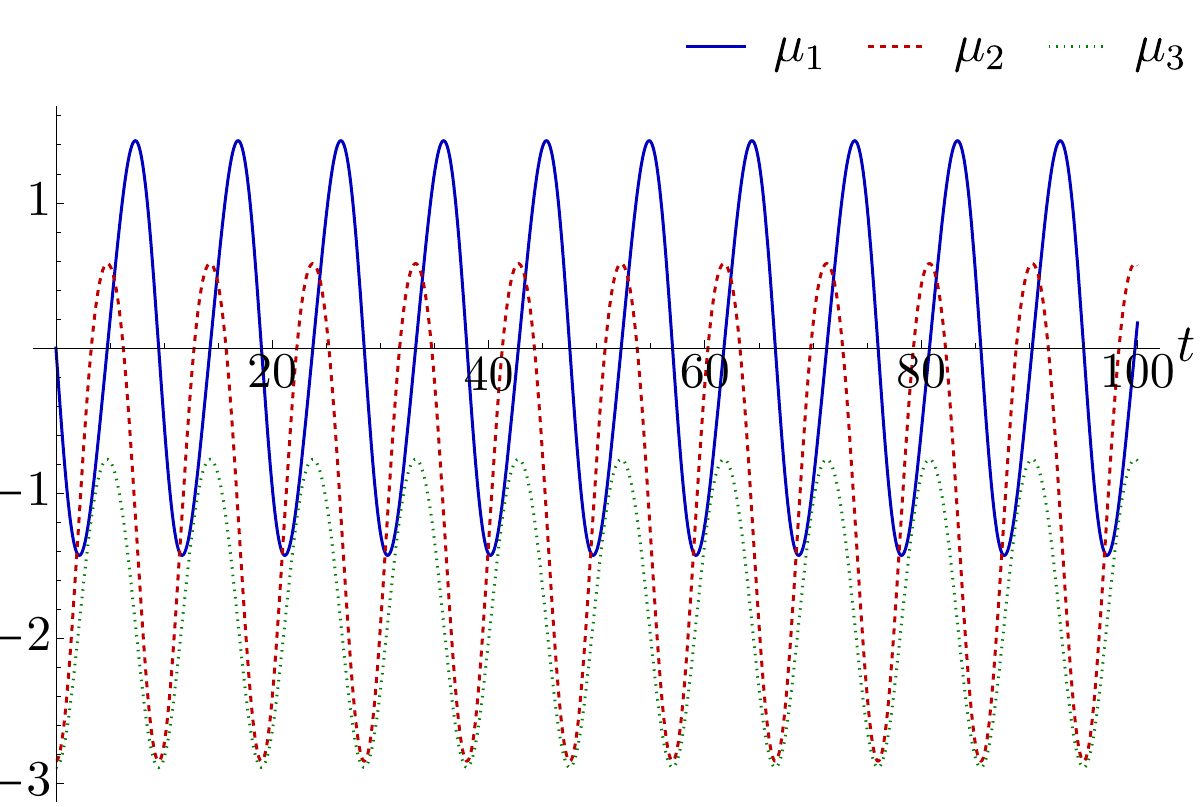}
  }
  \quad
  \subfigure[Lie--Poisson dynamics and invariants]{
    \includegraphics[width=.35\linewidth]{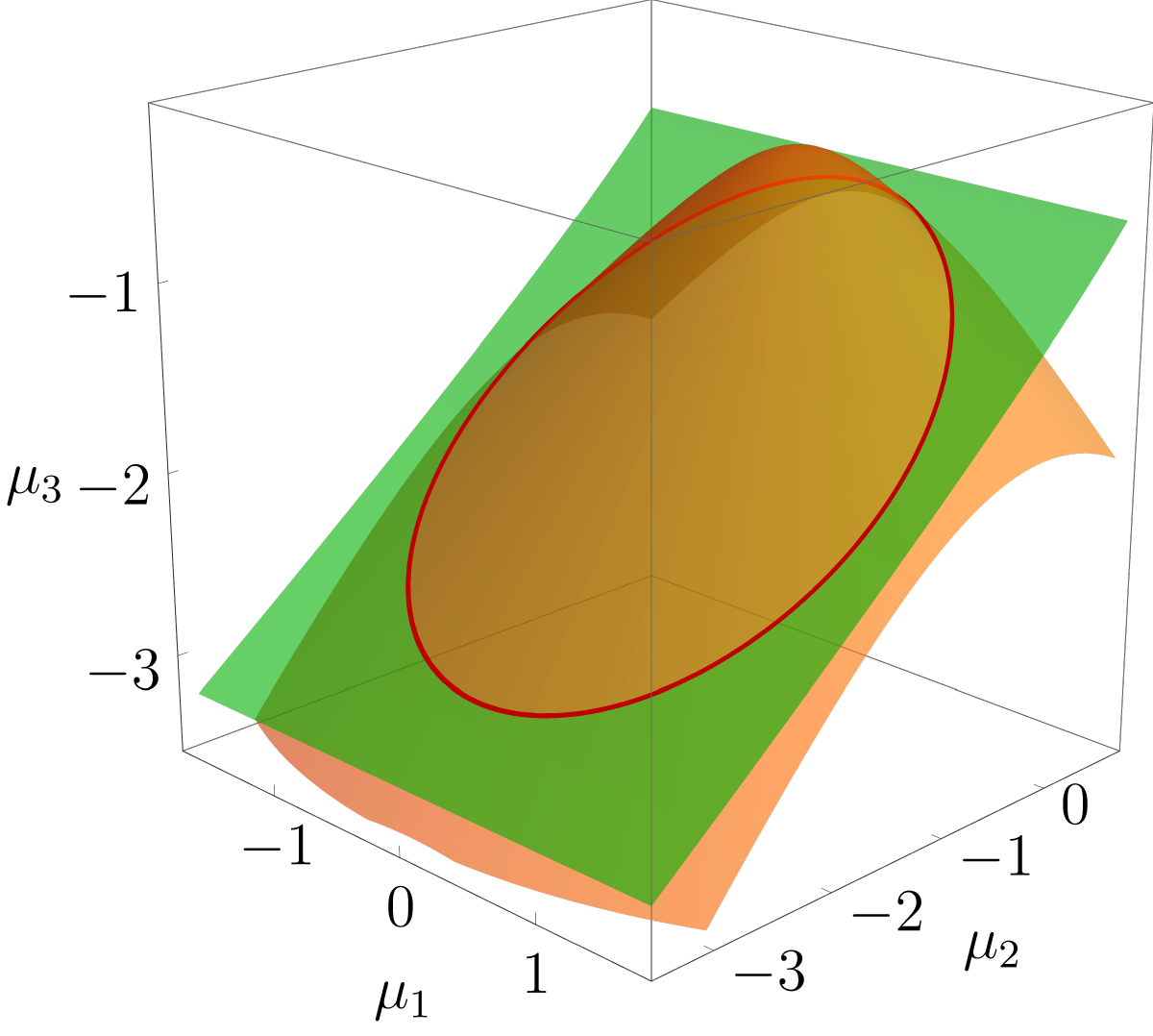}
  }
  \captionsetup{width=0.95\textwidth}
  \caption{
    (a)~Time evolution of $\mu$ computed using the canonized system~\eqref{eq:canonized-Kida}.
    The solutions are shown for the time interval $0 \le t \le 100$ with time step $\Delta t = 0.1$.
    (b)~The red curve is the Lie--Poisson dynamics of the Kida vortex in $\mathfrak{g}^{*} = \so(2,1)^{*} \cong \R^{3}$ computed using the canonized system~\eqref{eq:canonized-Kida} and mapped by $\mathbf{M}^{+}$ in \eqref{eq:M-Kida}.
    The green and orange surfaces are the level sets of the Hamiltonian $h$ and the Casimir $f_{1}$ from \eqref{eq:h-Kida} and \eqref{eq:Casimir-Kida}, respectively.
  }
  \label{fig:mu-Kida}
\end{figure}

For comparison, we also solved the Lie--Poisson equation~\eqref{eq:LP-Kida} directly using the $4^{\rm th}$ order explicit Runge--Kutta method.
\Cref{fig:RK4vsIRK4-Kida} compares the time evolutions of the relative errors in the Hamiltonian $h$ and the Casimir $f_{1}$ along these numerical solutions.
The explicit Runge--Kutta solution exhibits a drift that seems to be detrimental in the long run.
Notice also that it exhibits a more significant drift in the Casimir.
On the other hand, the solution of the collective integrator does not exhibit drifts in either the Hamiltonian or the Casimir; note that the latter is preserved exactly in theory.

\begin{figure}[htbp]
  \centering
  \subfigure[Hamiltonian $h$ from \eqref{eq:h-Kida}]{
    \includegraphics[width=.45\linewidth]{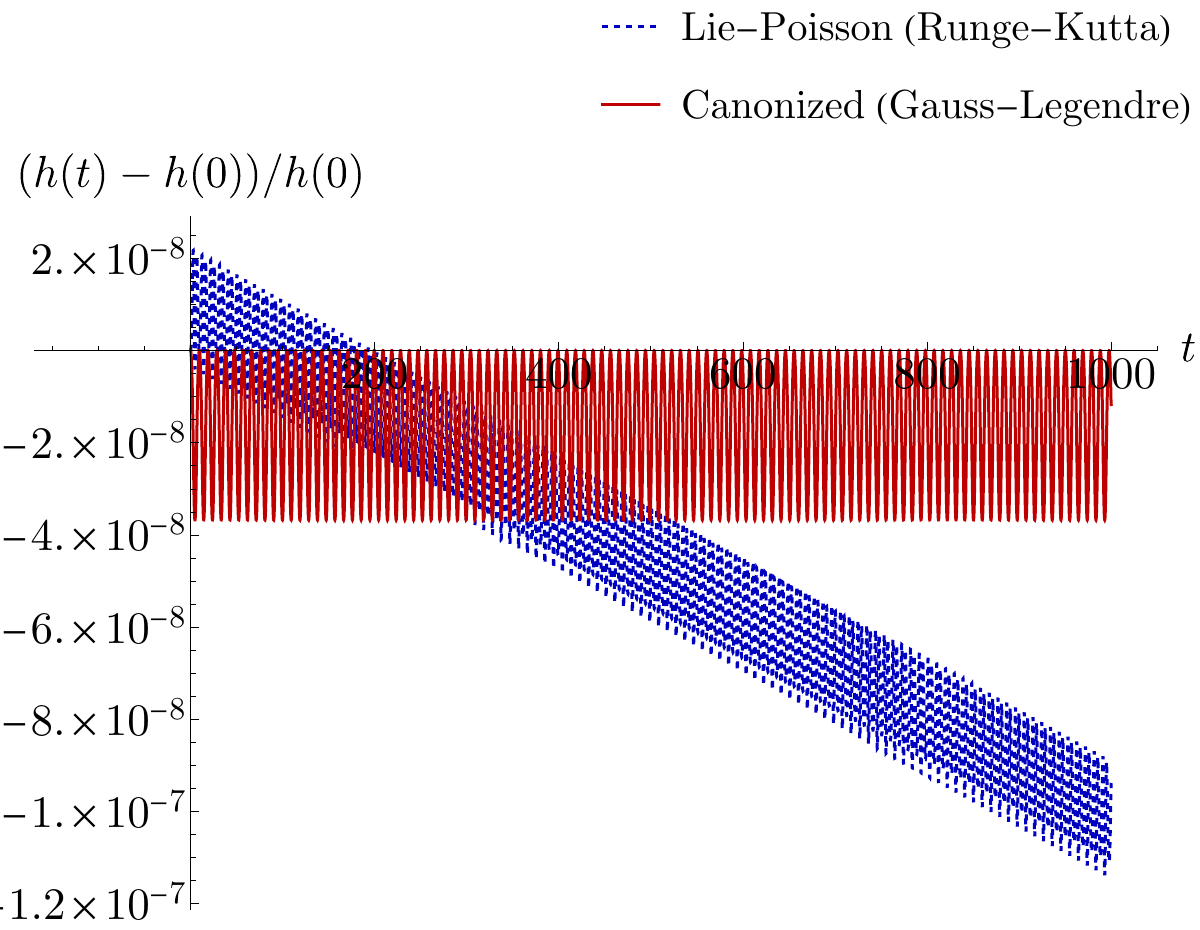}
  }
  \quad
  \subfigure[Casimir $f_{1}$ from \eqref{eq:Casimir-Kida}]{
    \includegraphics[width=.45\linewidth]{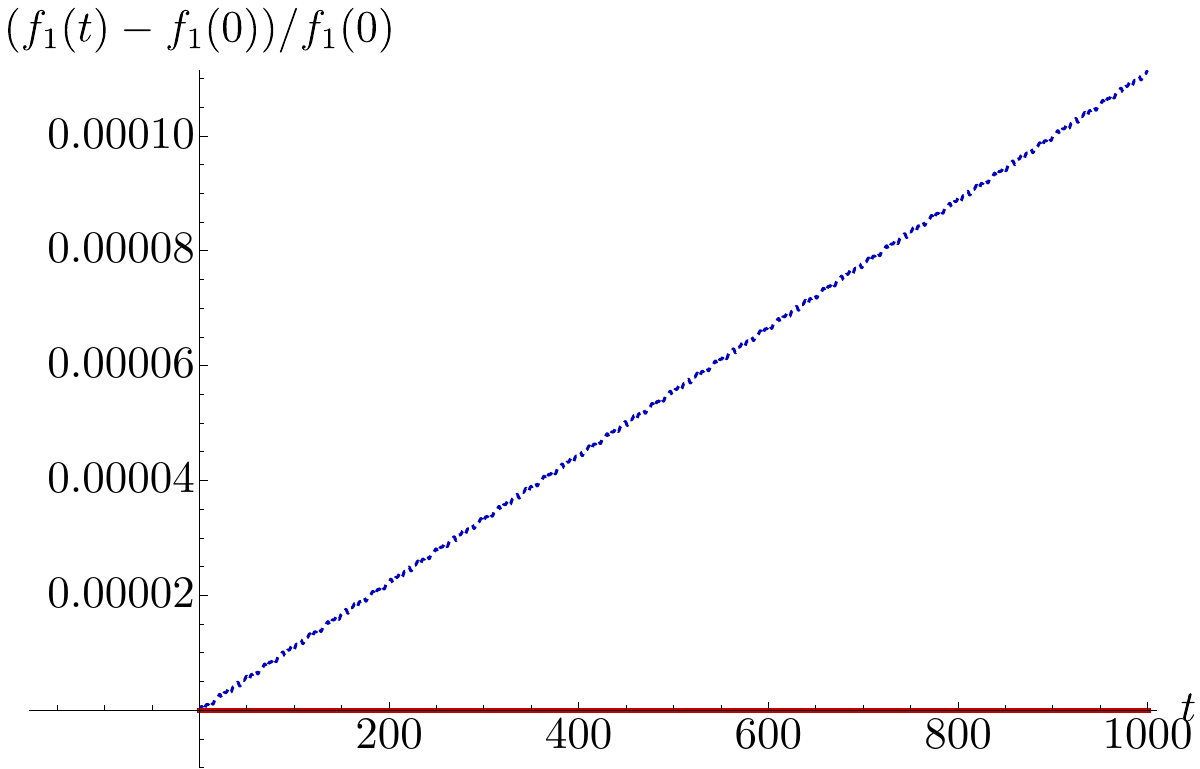}
  }
  \captionsetup{width=0.95\textwidth}
  \caption{
    Time evolutions of relative errors in Hamiltonian $h$ and Casimir $f_{1}$ from the Kida system.
    The dashed blue curve is the $4^{\rm th}$ order explicit Runge--Kutta method directly applied to Lie--Poisson equation~\eqref{eq:LP-Kida} whereas the solid red curve is the $4^{\rm th}$ order Gauss--Legendre method applied to the canonized system~\eqref{eq:canonized-Kida}.
    The solutions are shown for the time interval $0 \le t \le 1000$ with time step $\Delta t = 0.1$.
    Note that, in (b), the red line is made thicker to make it visible; the actual variation is so small that it is barely visible if plotted with the same thickness as the blue line or as in (a).
  }
  \label{fig:RK4vsIRK4-Kida}
\end{figure}

\Cref{fig:Killing_inv-Kida} shows how well the collective integrator preserves the components of the momentum map $\mathbf{J}$.
This is because the Gauss--Legendre methods preserve these invariants exactly in theory.
However, being an implicit method, it introduces an error in each step when solving nonlinear equations---the likely culprit of the small errors observed in the figures.

\begin{figure}[htbp]
  \centering
  \subfigure[Component $J_{0}$]{
    \includegraphics[width=.31\linewidth]{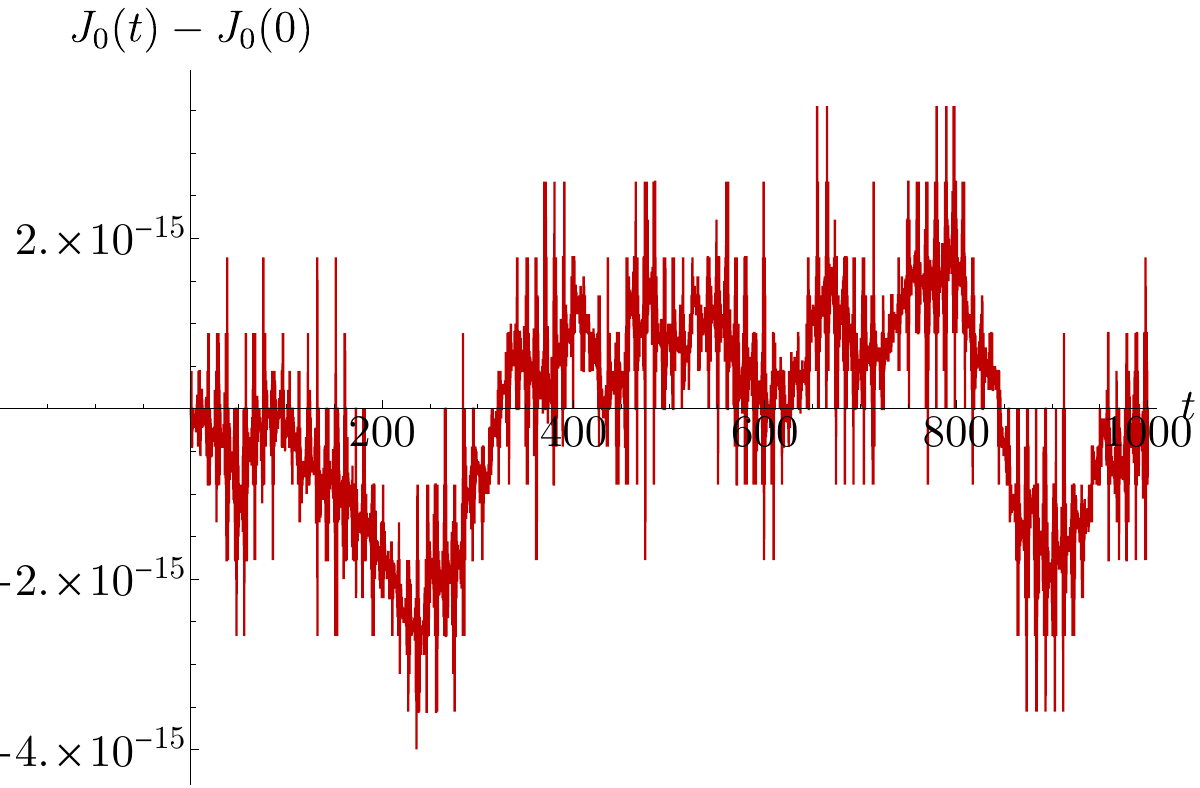}
  }
  \subfigure[Component $J_{1}$]{
    \includegraphics[width=.31\linewidth]{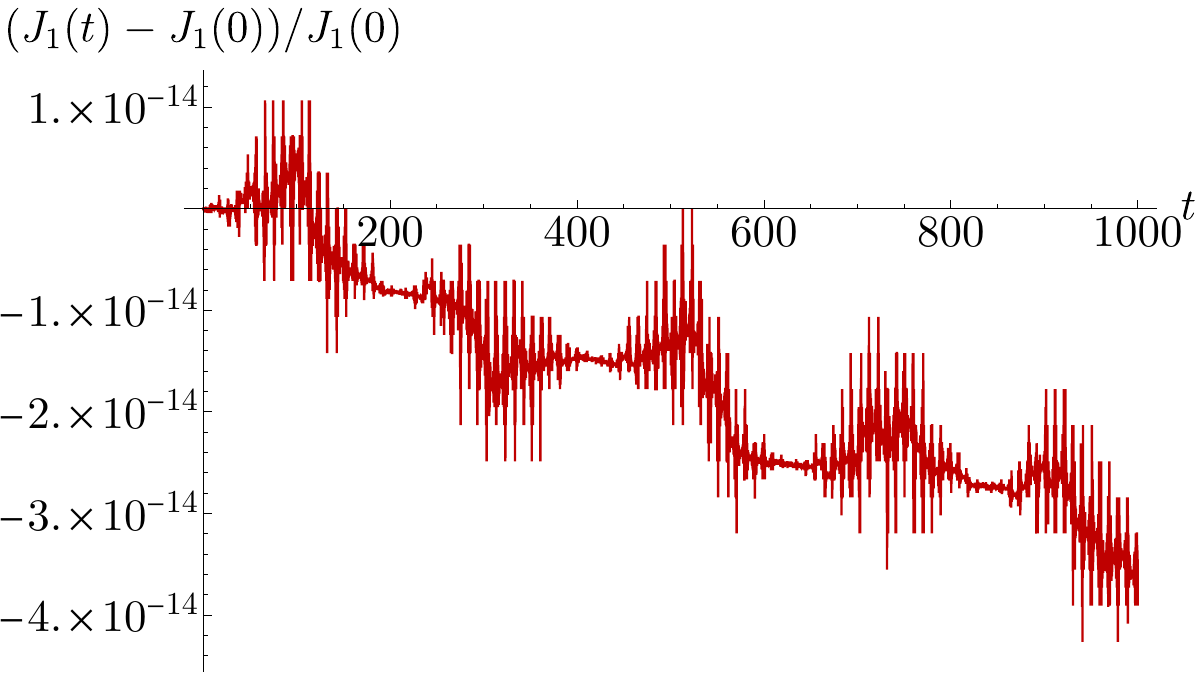}
  }
  \subfigure[Component $J_{2}$]{
    \includegraphics[width=.31\linewidth]{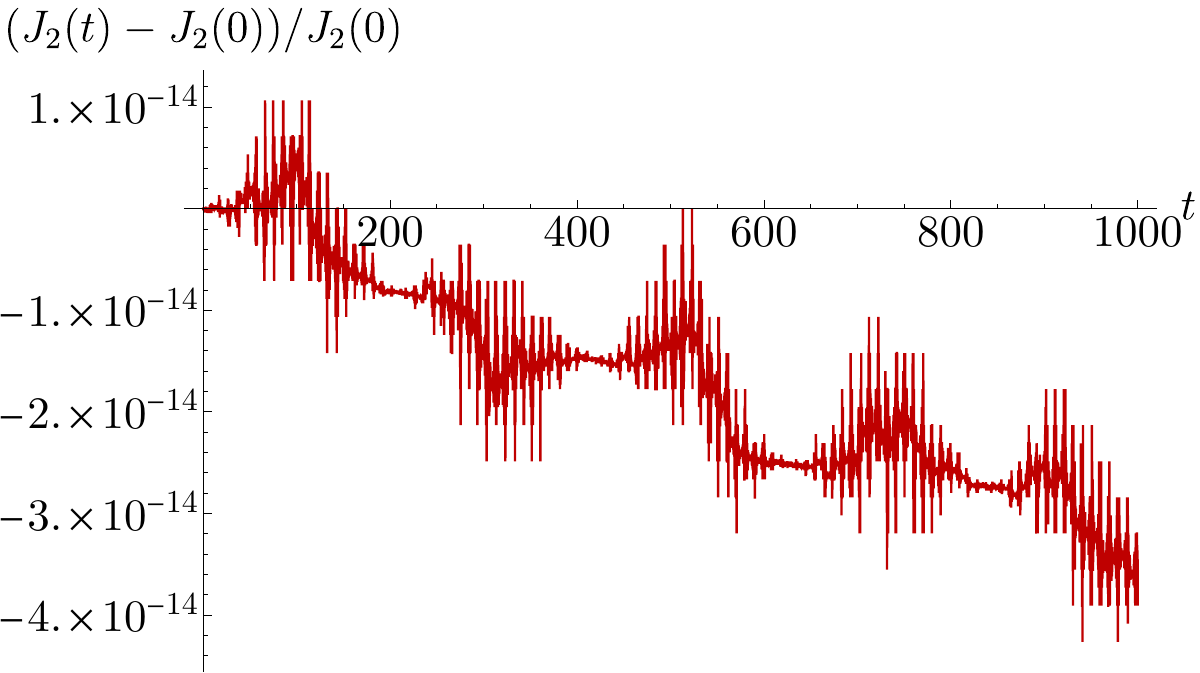}
  }
  \captionsetup{width=0.95\textwidth}
  \caption{
    Time evolutions of absolute or relative errors in components of momentum map $\mathbf{J}$ from  \eqref{eq:J_0} and  \eqref{eq:Killing_inv-Kida} computed by the $4^{\rm th}$ order Gauss--Legendre method applied to the canonized Kida system~\eqref{eq:canonized-Kida}.
    The solutions are shown for the time interval $0 \le t \le 1000$ with time step $\Delta t = 0.1$.
  }
  \label{fig:Killing_inv-Kida}
\end{figure}

\subsection{Example~2: Heavy Top on a Movable Base}
\label{ssec:htmb}
As a higher-dimensional and more practical example, consider the system shown in \Cref{fig:htmb-details} from \citet{CoOh-EPwithBSym1}: It is a heavy top with mass $m$ placed on a movable base---point mass $M$ for simplicity---under gravity $\mathrm{g}$.

As the base is free to move, the system is defined by the rotational motion of the heavy top and the linear motion of the base. Hence the natural configuration space is the matrix Lie group
\begin{equation*}
  \SE(3) = \setdef{ (R,\mathbf{x}) \defeq \begin{bmatrix} R & \mathbf{x} \\0 & 1 \end{bmatrix} }{ R \in \SO(3), \mathbf{x} \in  \R^{3}},
\end{equation*}
where $R \in \SO(3)$ gives the orientation of the top and $\mathbf{x}$ is the position of the base.
The left translation of the tangent vector $(\dot{R}, \dot{\mathbf{x}}) \in T_{(R,\mathbf{x})}\SE(3)$ to the identity yields
\begin{equation*}
  \begin{bmatrix}
    \hat{\mathbf{\Omega}} & \mathbf{v} \\
    0 & 0
  \end{bmatrix}
  \defeq
  \begin{bmatrix}
    R & \mathbf{x} \\
    0 & 1
  \end{bmatrix}^{-1}
  \begin{bmatrix}
    \dot{R} & \dot{\mathbf{x}} \\
    0 & 0
  \end{bmatrix}
  =
  \begin{bmatrix}
    R^{-1}\dot{R} & R^{-1}\dot{\mathbf{x}} \\
    0 & 0
  \end{bmatrix}
  \in \se(3),
\end{equation*}
which are the angular velocity of the top and the base velocity with respect to the body frame of the top.
Note that we identify $\se(3) = \so(3) \ltimes \R^{3}$ with $\R^{3} \times \R^{3}$ via the hat map $\hat{(\,\cdot\,)}\colon \R^{3} \to \so(3)$; see, e.g., \cite[Eq.~(9.2.7) on p.~289]{MaRa1999}.

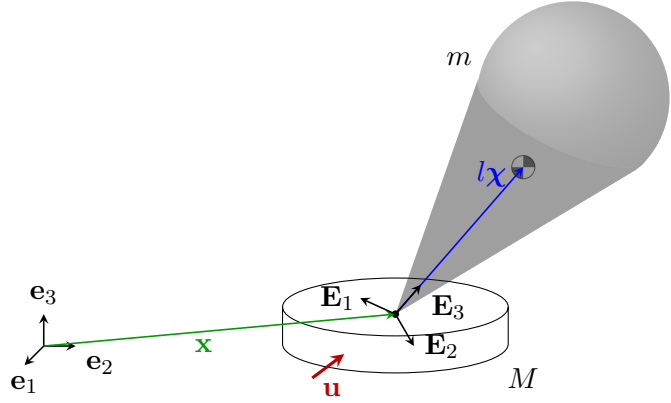
\begin{figure}[htbp]
  \centering
  \begin{tikzpicture}[scale=.85]
    \node [draw, cylinder, shape aspect=3, rotate=90, minimum height=1cm, minimum width=3cm] (c1)  at (0,1){};
    \draw[fill] (0,1.5) circle [radius=0.05];

    \fill[
    left color=gray!50!black,
    middle color=gray!50,
    right color=gray!50!black,
    shading=axis,
    opacity=0.25
    ] (3.75,3.75) -- (0,1.5) -- (1.3,5.2) [rotate=335] arc (180:340:1.49cm and 0.5cm); 
    \fill[
    shading=ball,
    ball color=gray!50,
    opacity=0.25
    ] (1.3,5.2) [rotate=335] arc (180:340:1.49cm and 0.5cm) arc (340:560:1.507cm); 

    \node at (1,5.5)    {$m$};
    \node at (2,0.5)    {$M$};

    \node[semitransparent] at (2,3.8) {\centerofmass};
    \draw[semithick,blue,->,>=stealth]  (0,1.5) -- (2,3.8);
    \node[left] at (1.9,3.7) {\color{blue}{$l\boldsymbol{\chi}$}};

    \draw[semithick,->,>=stealth]  (-5.5,1) -- (-5.5,1.5) node[above] {$\mathbf{e}_3$};
    \draw[semithick,->,>=stealth]  (-5.5,1) -- (-5.8,0.7) node[below] {$\mathbf{e}_1$};
    \draw[semithick,->,>=stealth]  (-5.5,1) -- (-5,1) node[below right] {$\mathbf{e}_2$};

    \draw[semithick,green!60!black][->,>=stealth]  (-5.5,1) -- (0,1.5);
    \node at (-3,1) {\color{green!60!black}{$\mathbf{x}$}};

    \draw[semithick,->,>=stealth]  (0,1.5) -- (-0.55,1.75) node[left] {$\mathbf{E}_1\!$};
    \draw[semithick,->,>=stealth]  (0,1.5) -- (0.29,1) node[right] {$\mathbf{E}_2$};
    \draw[semithick,->,>=stealth]  (0,1.5) -- (0.4,1.96) node[below right] {$\mathbf{E}_3$};

    \draw[very thick, ->,>=stealth, red!75!black] (-1.3,0.5) -- (-0.8,0.875);
    \node[below right, red!75!black] at (-1.3,0.575) {$\mathbf{u}$};

  \end{tikzpicture}
  \caption{ Heavy top on a movable base.}
  \label{fig:htmb-details}
\end{figure}

Let $\bar{m} \defeq m+M$ be the total mass of the system, $l$ the distance from the junction point of the top and the base to the center of mass of the heavy top, $\boldsymbol{\chi}$ the unit vector in that direction in the body frame, and $\mathbb{I}_{0} \defeq \diag(I_{1}, I_{2}, I_{3})$ the inertia mass matrix of the top with respect to the junction point  (we assume $I_{1} = I_{2}$); see \Cref{fig:htmb-details}.

Using the body angular momentum $\bPi$ and the linear impulse $\bP$ related to $\bOmega$ and $\mathbf{v}$ as
\begin{equation*}
  \mathbf{\Pi} = \mathbb{I}_{0} \mathbf{\Omega} + ml \boldsymbol{\chi} \times \mathbf{v},
  \qquad
  \mathbf{P} = - ml \boldsymbol{\chi} \times \mathbf{\Omega} + \bar{m} \mathbf{v},
\end{equation*}
the Hamiltonian of the system is
\begin{equation*}
  h(\bPi, \bP, \bGamma, x_{3}) \defeq
  \dfrac{1}{2} \parentheses{
    \mathbf{\Pi} \cdot (\mathbb{I}^{-1} \mathbf{\Pi}) + 2k m l \mathbf{\Pi} \cdot (\mathbf{P} \times \boldsymbol{\chi}) + \mathbf{P} \cdot (\mathbb{M}^{-1} \mathbf{P})
  } + m \mathrm{g} l \boldsymbol{\chi} \cdot \mathbf{\Gamma}
  + \bar{m} \mathrm{g} x_{3}
\end{equation*}
with
\begin{equation*}
  \mathbb{I} \defeq \diag\parentheses{ I_{1} - \frac{m^{2}l^{2}}{\bar{m}},\, I_{1} - \frac{m^{2}l^{2}}{\bar{m}},\, I_{3} },
  \qquad
  \mathbb{M} \defeq \diag\parentheses{\bar{m} - \frac{m^{2}l^{2}}{I_{1}},\, \bar{m} - \frac{m^{2}l^{2}}{I_{1}},\, \bar{m} }.
\end{equation*}
Then the equations of motion are written as the Lie--Poisson equation on $(\se(3) \ltimes \R^{4})^{*}$:
\begin{equation}
  \label{eq:LP-htmb}
  \begin{array}{c}
    \DS \dot{\bPi} = \bPi \times \pd{h}{\bPi} + \bP \times \pd{h}{\bP} + \bGamma \times \pd{h}{\bGamma},
    \qquad
    \DS \dot{\bP} = \bP \times \pd{h}{\bPi} - \pd{h}{x_{3}}\bGamma,
    \smallskip\\
    \DS
    \dot{\bGamma} = \bGamma \times \pd{h}{\bPi},
    \qquad
    \dot{x}_{3} = \bGamma \cdot \pd{h}{\bP}.
  \end{array}
\end{equation}
The main goal of \cite{CoOh-EPwithBSym1} is to stabilize the upright position of the heavy top by applying control $\mathbf{u}$ to the base, i.e., the second equation of \eqref{eq:LP-htmb} is replaced by
\begin{equation*}
  \dot{\bP} = \bP \times \pd{h}{\bPi} - \pd{h}{x_{3}}\bGamma + \mathbf{u},
\end{equation*}
Specifically, the control $\mathbf{u}$ was broken into two as $\mathbf{u} = \mathbf{u}^{\rm p} + \mathbf{u}^{\rm k}$, corresponding to the potential and kinetic shaping, with the potential part being $\mathbf{u}^{\rm p} = \pd{h}{x_{3}} \mathbf{\Gamma} = \bar{m} \mathrm{g} \mathbf{\Gamma}$, so that the Lie--Poisson equation~\eqref{eq:LP-htmb} now becomes
\begin{equation}
  \label{eq:CLP-htmb}
  \dot{\bPi} = \bPi \times \pd{h}{\bPi} + \bP \times \pd{h}{\bP} + \bGamma \times \pd{h}{\bGamma},
  \qquad
  \dot{\bP} = \bP \times \pd{h}{\bPi} + \mathbf{u}^{\rm k},
  \qquad
  \dot{\bGamma} = \bGamma \times \pd{h}{\bPi},
\end{equation}
where we dropped the equation for $x_{3}$ because it is now decoupled from the rest.
In \cite{CoOh-EPwithBSym1}, it is found, via the method of controlled Lagrangians~\cite{BlLeMa2000,BlLeMa2001}, applying the control
\begin{equation*}
  \mathbf{u}^{\rm k} = (\rho - \bar{m}) \parentheses{ \dot{\mathbf{v}} - \mathbf{v} \times \mathbf{\Omega} }
  \text{ where }
  \mathbf{v} \defeq \pd{h}{\bP}
\end{equation*}
with $\rho \in \R$ renders the system~\eqref{eq:CLP-htmb} the Lie--Poisson equation on $\parentheses{\se(3) \ltimes \R^{3}}^{*}$ with a new control Hamiltonian $h_{\rm c}\colon (\se(3) \ltimes \R^{3})^{*} \to \R$ given by
\begin{equation}
  \label{eq:h_c}
    h_{\rm c}(\mathbf{\Pi}, \mathbf{P}, \mathbf{\Gamma})
    = \dfrac{1}{2} \parentheses{
      \mathbf{\Pi} \cdot (\mathbb{I}_{\rm c}^{-1} \mathbf{\Pi})
      + 2k_{\rm c} m l \mathbf{\Pi} \cdot (\mathbf{P} \times \boldsymbol{\chi})
      + \mathbf{P} \cdot (\mathbb{M}_{\rm c}^{-1} \mathbf{P})
    } + m \mathrm{g} l \boldsymbol{\chi} \cdot \mathbf{\Gamma}
\end{equation}
with
\begin{equation*}
  \mathbb{I}_{\rm c} \defeq \diag\parentheses{ I_{1} - \frac{m^{2}l^{2}}{\rho},\, I_{1} - \frac{m^{2}l^{2}}{\rho},\, I_{3} },
  \qquad
  \mathbb{M}_{\rm c} \defeq \diag\parentheses{ \rho - \frac{m^{2}l^{2}}{I_{1}},\, \rho - \frac{m^{2}l^{2}}{I_{1}},\, \rho }
\end{equation*}
Then the equations of motion are given by the Lie--Poisson equation
\begin{equation*}
  \dot{\mu} = \PB{\mu}{h_{\rm c}}_{-}
\end{equation*}
with $\mu = (\mathbf{\Pi},\mathbf{P},\mathbf{\Gamma}) \in \parentheses{ \se(3) \ltimes \R^{3} }^{*}$ and the following $(-)$-Lie--Poisson bracket on $\parentheses{ \se(3) \ltimes \R^{3} }^{*}$: For any smooth $f,g\colon (\se(3) \ltimes \R^{3})^{*} \to \R$,
\begin{equation} 
  \label{eq:CLPB-htmb}
  \begin{split}
    \PB{f}{g}_{-}(\mathbf{\Pi},\mathbf{P},\mathbf{\Gamma}) &= 
    - \anglebrackets{ \mathbf{\Pi}, \pd{f}{\mathbf{\Pi}} \times \pd{g}{\mathbf{\Pi}} } - \anglebrackets{ \mathbf{P}, \pd{f}{\mathbf{\Pi}} \times \pd{g}{\mathbf{P}} - \pd{g}{\mathbf{\Pi}} \times \pd{f}{\mathbf{P}} } \\
    &\quad - \anglebrackets{ \mathbf{\Gamma}, \pd{f}{\mathbf{\Pi}} \times \pd{g}{\mathbf{\Gamma}} - \pd{g}{\mathbf{\Pi}} \times \pd{f}{\mathbf{\Gamma}} },
  \end{split}
\end{equation}
which, incidentally, is identical to the Lie--Poisson bracket  given in  \citet{ThMo1998} for a rigid body insulator that is acted on by an electric field as well as gravity (see also \citet{ThMo2000, ThMo2001}).
More explicitly, we have
\begin{equation}
  \label{eq:CLP-htmb2}
  \dot{\bPi} = \bPi \times \pd{h_{\rm c}}{\bPi} + \bP \times \pd{h_{\rm c}}{\bP} + \bGamma \times \pd{h_{\rm c}}{\bGamma},
  \qquad
  \dot{\bP} = \bP \times \pd{h_{\rm c}}{\bPi},
  \qquad
  \dot{\bGamma} = \bGamma \times \pd{h_{\rm c}}{\bPi}.
\end{equation}

Noting that \eqref{eq:CLPB-htmb} is a $(-)$-Lie--Poisson bracket~\eqref{eq:LPB-}, we find that the corresponding structure constants $\{c^{k}_{ij}\}_{1\le i,j,k \le 9}$ satisfy
\begin{equation*}
  \mu_{k} c^{k}_{ij} =
  -\begin{bmatrix}
    \hat{\mathbf{\Pi}} & \hat{\mathbf{P}} & \hat{\mathbf{\Gamma}} \\
    \hat{\mathbf{P}} & 0 & 0  \\
    \hat{\mathbf{\Gamma}} & 0 & 0 
  \end{bmatrix}.
\end{equation*}
One can also show that the Lie--Poisson bracket~\eqref{eq:CLPB-htmb} possesses the following Casimirs:
\begin{equation}
  \label{eq:Casimirs-htmb}
  f_{1} = \norm{\mathbf{P}}^{2},
  \qquad
  f_{2} = \mathbf{P} \cdot \mathbf{\Gamma},
  \qquad
  f_{3} = \norm{\mathbf{\Gamma}}^{2}.
\end{equation}
Furthermore, we can write the momentum map $\mathbf{M}^{-}$ as
\begin{equation*}
  \mathbf{M}^{-}(q, p) = -(
  \mathbf{q}_{1} \times \mathbf{p}_{1} + \mathbf{q}_{2} \times \mathbf{p}_{2} + \mathbf{q}_{3} \times \mathbf{p}_{3},\,
  \mathbf{q}_{1} \times \mathbf{p}_{2},\,
  \mathbf{q}_{1} \times \mathbf{p}_{3}
  ),
\end{equation*}
where we used the identification $\mathfrak{g} = \se(3) \ltimes \R^{3} \cong \R^{3} \times \R^{3} \times \R^{3}$ and wrote $q = (\mathbf{q}_{1}, \mathbf{q}_{2}, \mathbf{q}_{3}) \in \mathfrak{g} \cong \R^{9}$ and $p = (\mathbf{p}_{1}, \mathbf{p}_{2}, \mathbf{p}_{3}) \in \mathfrak{g}^{*} \cong \R^{9}$ with $\mathbf{q}_{i}, \mathbf{p}_{i} \in \R^{3}$ for $i \in \{1,2,3\}$.
Defining the Hamiltonian $H\colon T^{*}(\se(3) \ltimes \R^{3}) \to \R$ as $H(q,p) = h_{\rm c}(\mathbf{M}^{-}(q,p))$, we have the canonized system~\eqref{eq:canonized}.

Let us find the other momentum map (invariant) $\mathbf{J}$.
Using \eqref{eq:mathfrak_h-characterization}, we find that $\mathfrak{h}$ is the 9-dimensional subalgebra of $\sym(18,\R)$ consisting of matrices of the form
\begin{equation*}
  \sigma =
  \begin{bmatrix}
    \sigma_{11} & \sigma_{12} \\
    \sigma_{12}^{T} & \sigma_{22}
  \end{bmatrix}
  \in \sym(18,\R)
\end{equation*}
with
\begin{gather*}
  \sigma_{12} \in \Span\left\{
    I_{9},
    \begin{bmatrix}
      0 & I_{3} & 0 \\
      0 & 0 & 0 \\
      0 & 0 & 0
    \end{bmatrix},
    \begin{bmatrix}
      0 & 0 & I_{3} \\
      0 & 0 & 0 \\
      0 & 0 & 0 \\
    \end{bmatrix}
  \right\},
  \\
  \sigma_{11} \in \Span\left\{
  \begin{bmatrix}
    I_{3} & 0 & 0 \\
    0 & 0 & 0 \\
    0 & 0 & 0
  \end{bmatrix},
  \begin{bmatrix}
    0 & I_{3} & 0 \\
    I_{3} & 0 & 0 \\
    0 & 0 & 0
  \end{bmatrix},
  \begin{bmatrix}
    0 & 0 & I_{3} \\
    0 & 0 & 0 \\
    I_{3} & 0 & 0
  \end{bmatrix}
  \right\},
  \\
  \sigma_{22} \in \Span\left\{
  \begin{bmatrix}
    0 & 0 & 0 \\
    0 & I_{3} & 0 \\
    0 & 0 & 0
  \end{bmatrix},
  \begin{bmatrix}
    0 & 0 & 0 \\
    0 & 0 & I_{3} \\
    0 & I_{3} & 0 \\
  \end{bmatrix},
  \begin{bmatrix}
    0 & 0 & 0 \\
    0 & 0 & 0 \\
    0 & 0 & I_{3}
  \end{bmatrix}
  \right\}.
\end{gather*}
Hence the components of the momentum map $\mathbf{J}$ are
\begin{equation}
  \label{eq:J-htmb}
  \begin{array}{lll}
    J_{0} \defeq p \cdot q, & J_{1} \defeq q_{1} p_{4} + q_{2} p_{5} + q_{3} p_{6}, & J_{2} \defeq q_{1} p_{7} + q_{2} p_{8} + q_{3} p_{9}, \medskip\\
    J_{3} \defeq q_{1}^{2} + q_{2}^{2} + q_{3}^{2}, & J_{4} \defeq q_{1} q_{4} + q_{2} q_{5} + q_{3} q_{6}, & J_{5} \defeq q_{1} q_{7} + q_{2} q_{8} + q_{3} q_{9}, \medskip\\
    J_{6} \defeq p_{4}^{2} + p_{5}^{2} + p_{6}^{2}, & J_{7} \defeq p_{4} p_{7} + p_{5} p_{8} + p_{6} p_{9}, & J_{8} \defeq p_{7}^{2} + p_{8}^{2} + p_{9}^{2}.
  \end{array}
  \quad
\end{equation}

We may now follow the proof of \Cref{thm:dual_pair} in \Cref{sec:proof-dual_pair} to show that there exists an open set $U$ that is dense in $T^{*}\mathfrak{g}$ on which $\mathbf{M}^{-}$ and $\mathbf{J}$ are submersions.
We also saw above that $\dim\mathfrak{h} = 9 = \dim\mathfrak{g}$.
Hence we have a dual pair as described in \Cref{thm:dual_pair}.

Following~\cite{CoOh-EPwithBSym1}, the parameters are chosen as follows: $M=0.44\,\mathrm{[kg]}$, $m=0.7\,\mathrm{[kg]}$, $I_{1} = I_{2} = 0.2\,\mathrm{kg \cdot m^{2}}$, $I_{3} = 0.24\,\mathrm{kg \cdot m^{2}}$, $l = 0.215\,\mathrm{[m]}$, $\mathrm{g} = 9.8\,\mathrm{[m/s^{2}]}$.
The parameter $\rho$ was chosen such that $\rho = 0.9m^{2}l^{2}/I_{1}$ to ensure stability of the upright position.
The initial condition is $\mathbf{\Omega}(0) = (0.1, 0.2, 0.1), \mathbf{v}(0) = \mathbf{0}$, and $\mathbf{\Gamma}(0) = (\cos \theta_{0} \sin \varphi_{0}, \sin \theta_{0} \sin \varphi_{0}, \cos \varphi_{0})$ with $\theta_{0} = \pi/3$ and $\varphi_{0} = \pi/20$.

To get the initial conditions for the canonized system, we set $\mathbf{q}_{1}(0) = \mathbf{\Gamma}(0) \times \mathbf{P}(0), \mathbf{p}_{1}(0) = (0,0,0)$ and solved $\mathbf{M}^{-} (q(0),p(0))=(\mathbf{\Pi}(0),\mathbf{P}(0),\mathbf{\Gamma}(0))$ for the remaining values $\mathbf{q}_{2}(0)$, $\mathbf{q}_{3}(0)$, $\mathbf{p}_{2}(0)$, $\mathbf{p}_{3}(0)$ of $(q(0),p(0))$.

We solved the canonized system using the $4^{\rm th}$ order Gauss--Legendre method, and also solved the Lie--Poisson system~\eqref{eq:CLP-htmb2} directly using the $4^{\rm th}$ order explicit Runge--Kutta method for comparison.

\Cref{fig:RK4vsIRK4-htmb} shows the time evolutions of the relative errors of the Hamiltonian $h_{\rm c}$ and the Casimirs $f_{1}, f_{2}, f_{3}$.
Just as in the Kida vortex case, we observe drifts in addition to oscillations in all the invariants for the explicit Runge--Kutta solution, whereas we see that the proposed collective integrator preserves these invariants: the Hamiltonian oscillates in a thin band, whereas the Casimirs are preserved exactly in theory.

The errors for the components of the momentum map $\mathbf{J}$ are shown in \Cref{fig:Killing_inv-htmb}.
Since all of them are quadratic in $(q,p)$, they are invariants of the Gauss--Legendre method, any error must be due to roundoff and/or the nonlinear solver used in each step.

\begin{figure}[htbp]
  \centering
  \subfigure[Controlled Hamiltonian $h_{\rm c}$ from \eqref{eq:h_c}]{
    \includegraphics[width=.45\linewidth]{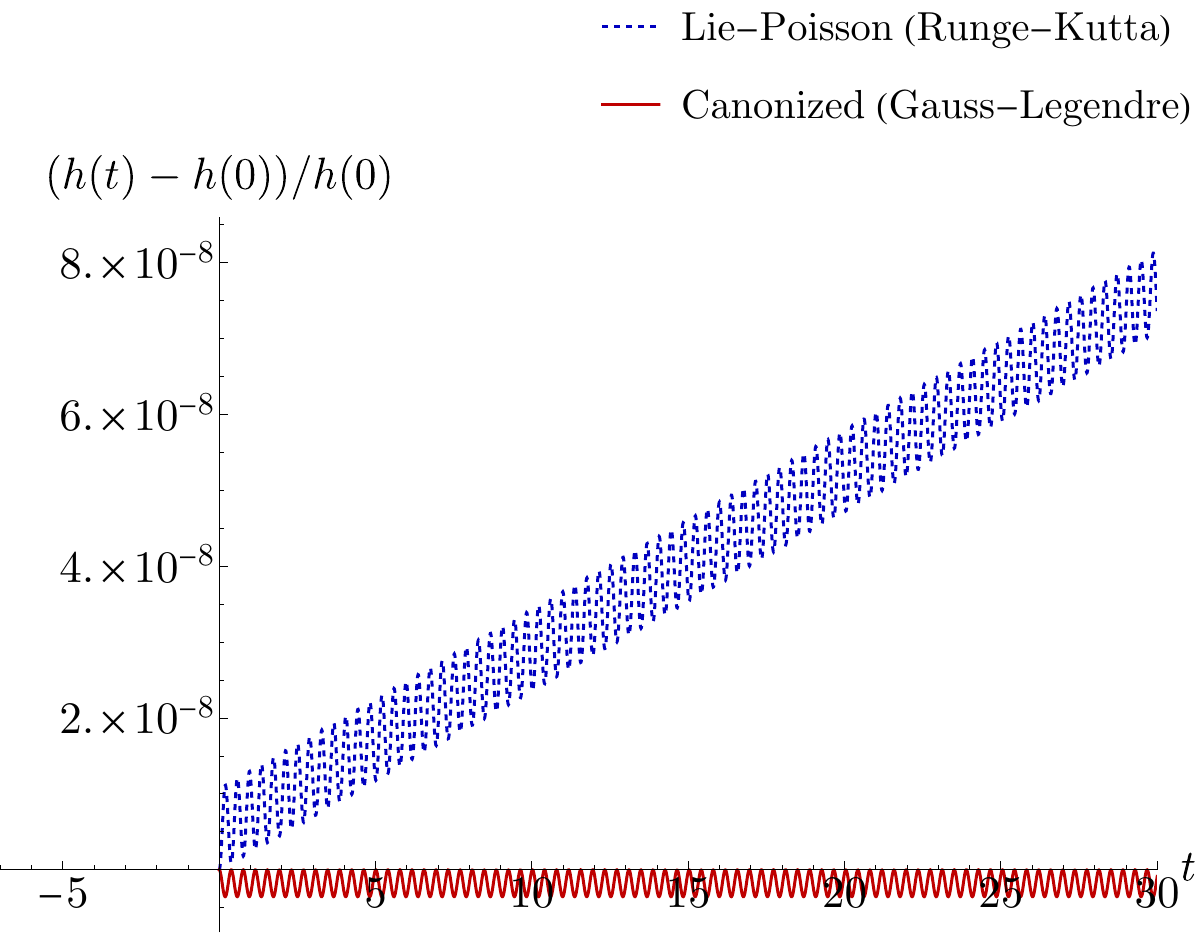}
  }
  \quad
  \subfigure[Casimir $f_{1}$ from \eqref{eq:Casimirs-htmb}]{
    \includegraphics[width=.45\linewidth]{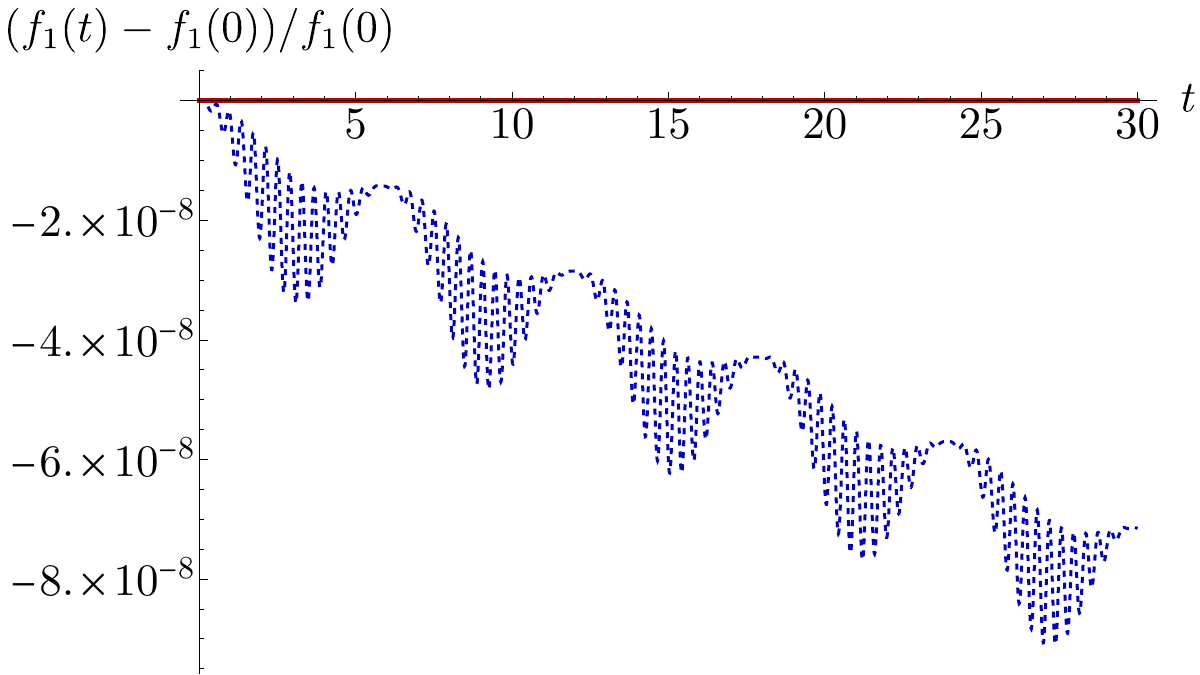}
  }
  \quad
  \subfigure[Casimir $f_{2}$ from \eqref{eq:Casimirs-htmb}]{
    \includegraphics[width=.45\linewidth]{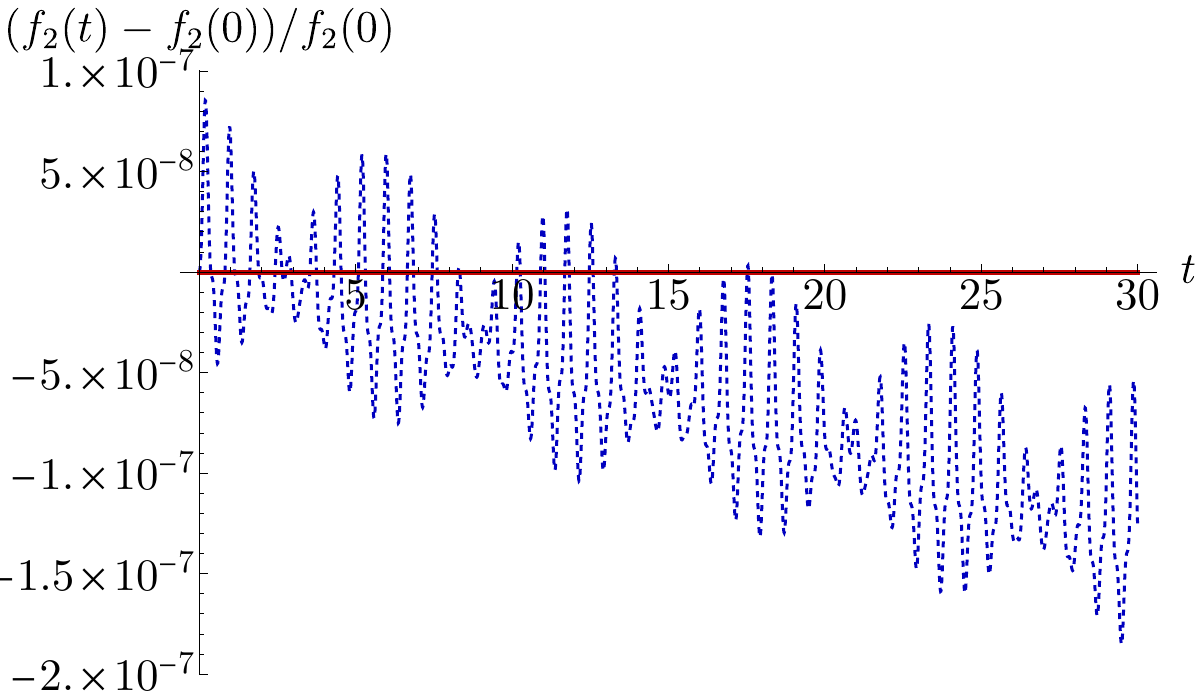}
  }
  \quad
  \subfigure[Casimir $f_{3}$ from \eqref{eq:Casimirs-htmb}]{
    \includegraphics[width=.45\linewidth]{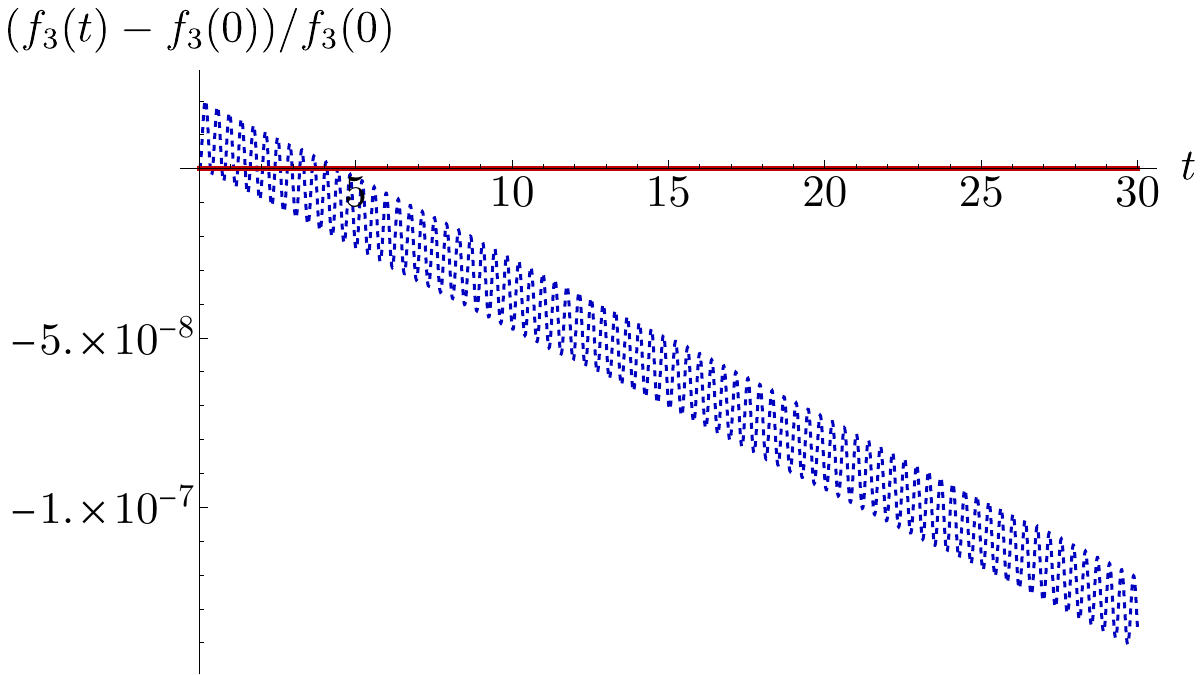}
  }
  \captionsetup{width=0.95\textwidth}
  \caption{
    Time evolutions of relative errors in Hamiltonian $h$ and three Casimirs $f_{1}$, $f_{2}$, $f_{3}$ from the heavy top on a movable base system.
    The dashed blue curve is the Runge--Kutta method directly applied to Lie--Poisson equation~\eqref{eq:CLP-htmb2} whereas the solid red curve is the $4^{\rm th}$ order Gauss--Legendre method applied to the canonized system.
    The solutions are shown for the time interval $0 \le t \le 30$ with time step $\Delta t = 0.01$.
    Note that, in (b)--(d), the red line is made thicker to make it visible; the actual variation is so small that it is barely visible if plotted with the same thickness as the blue line or as in (a).
  }
  \label{fig:RK4vsIRK4-htmb}
\end{figure}

\begin{figure}[htbp]
  \centering
  \subfigure[Component $J_{0}$]{
    \includegraphics[width=.31\linewidth]{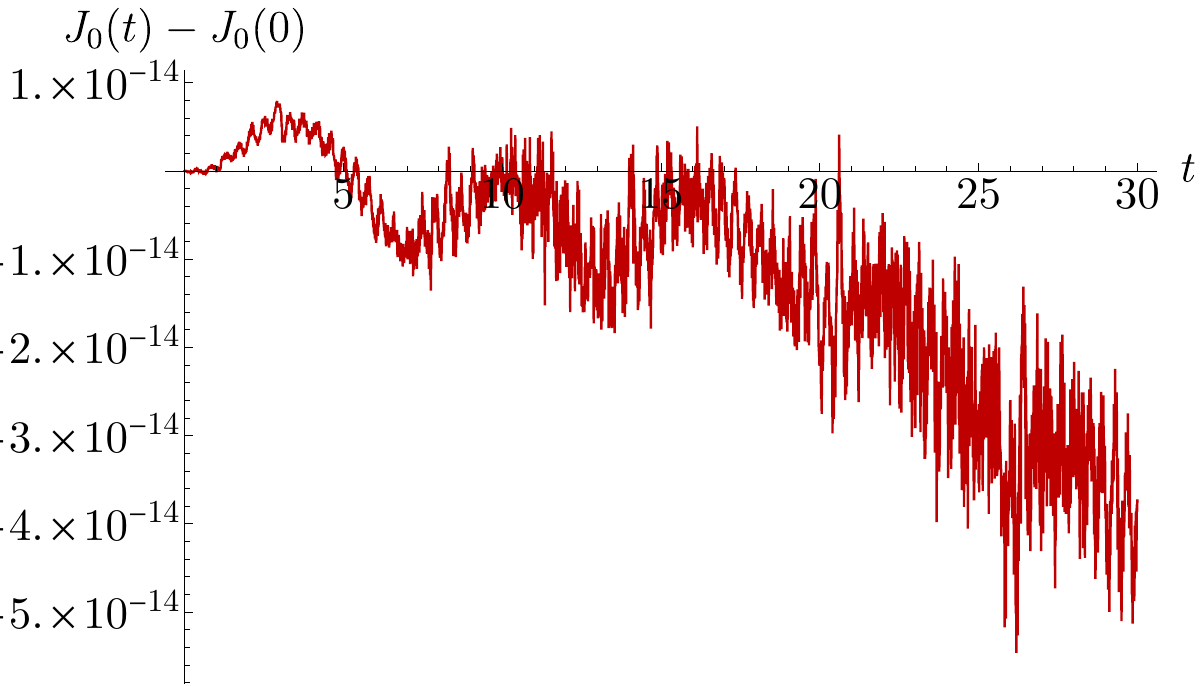}
  }
  \subfigure[Component $J_{1}$]{
    \includegraphics[width=.31\linewidth]{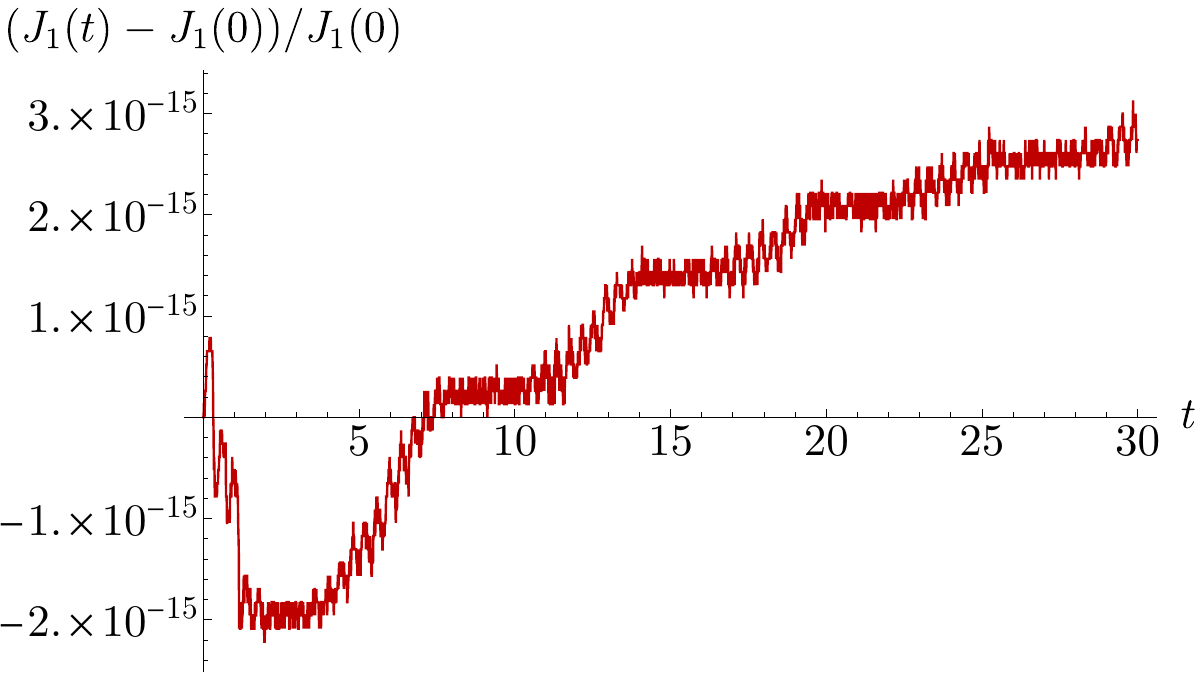}
  }
  \subfigure[Component $J_{2}$]{
    \includegraphics[width=.31\linewidth]{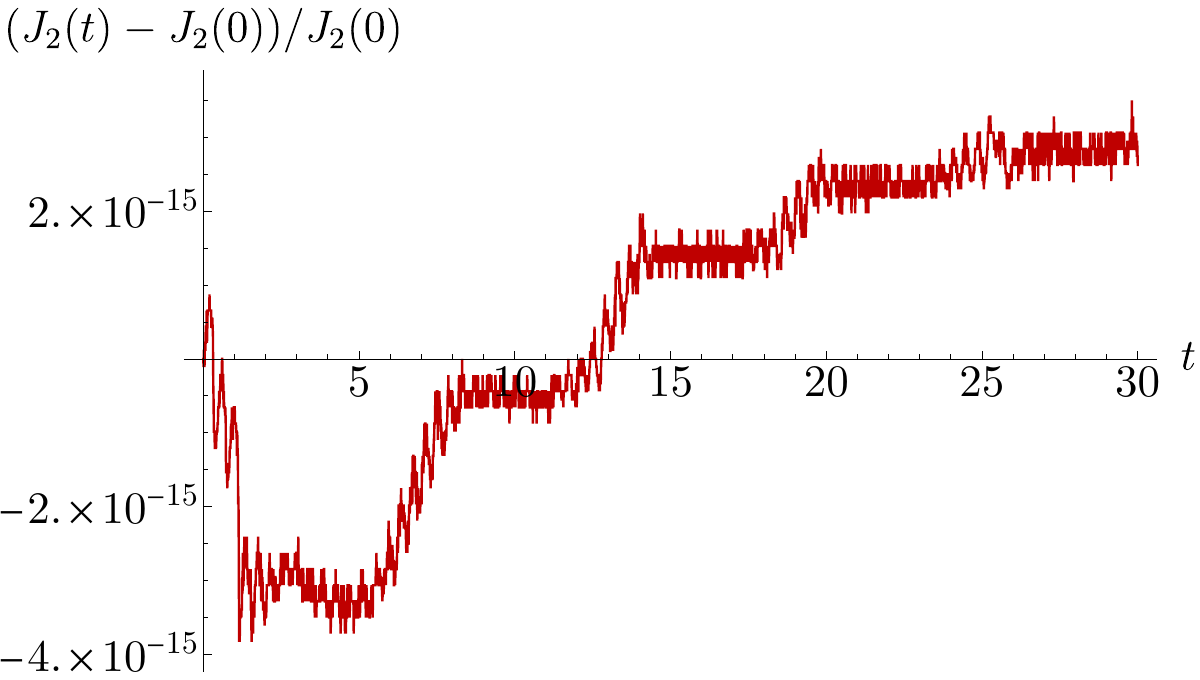}
  }
  \\
  \subfigure[Component $J_{3}$]{
    \includegraphics[width=.31\linewidth]{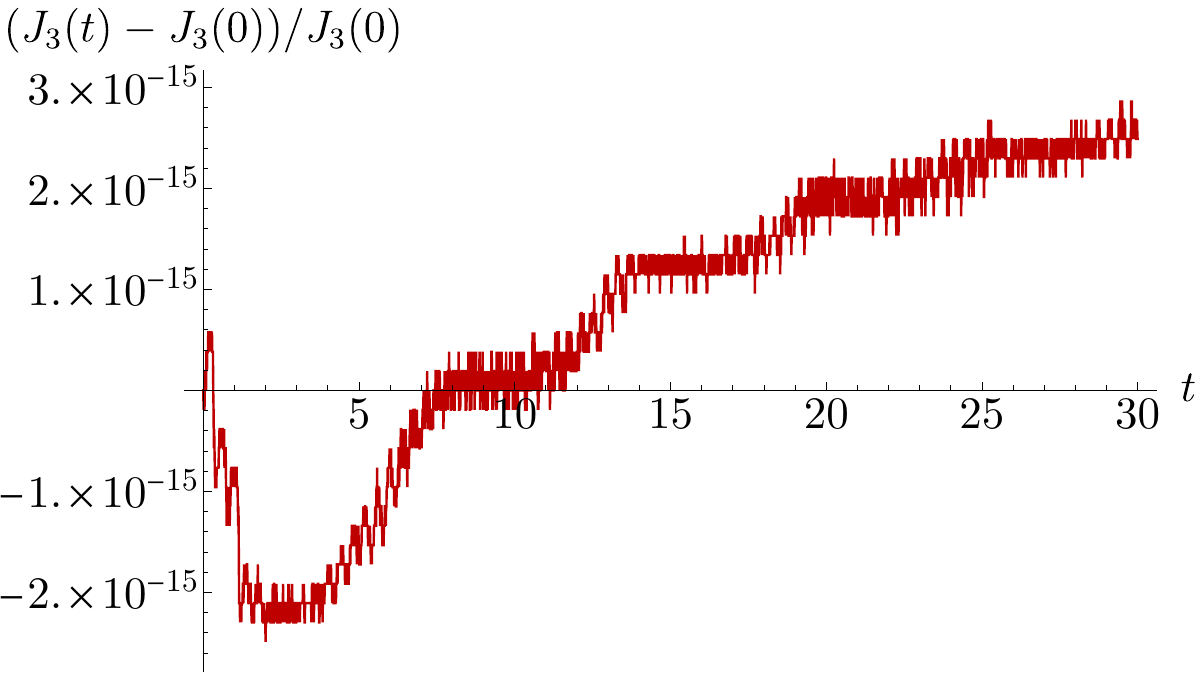}
  }
  \subfigure[Component $J_{4}$]{
    \includegraphics[width=.31\linewidth]{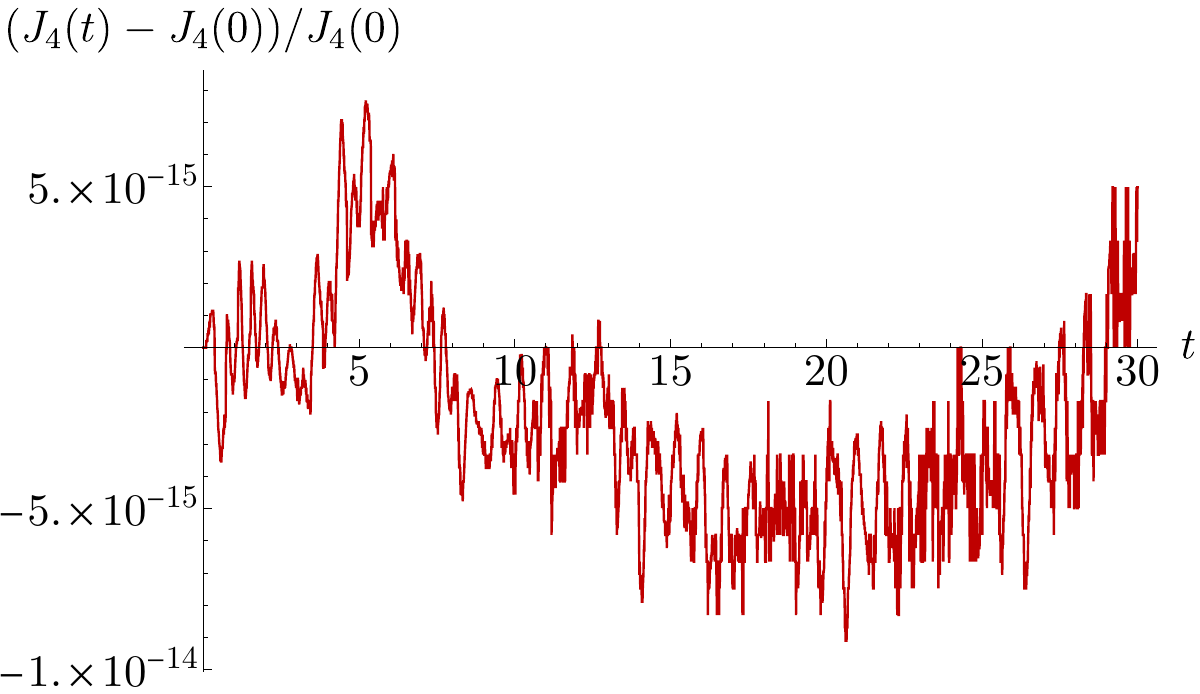}
  }
  \subfigure[Component $J_{5}$]{
    \includegraphics[width=.31\linewidth]{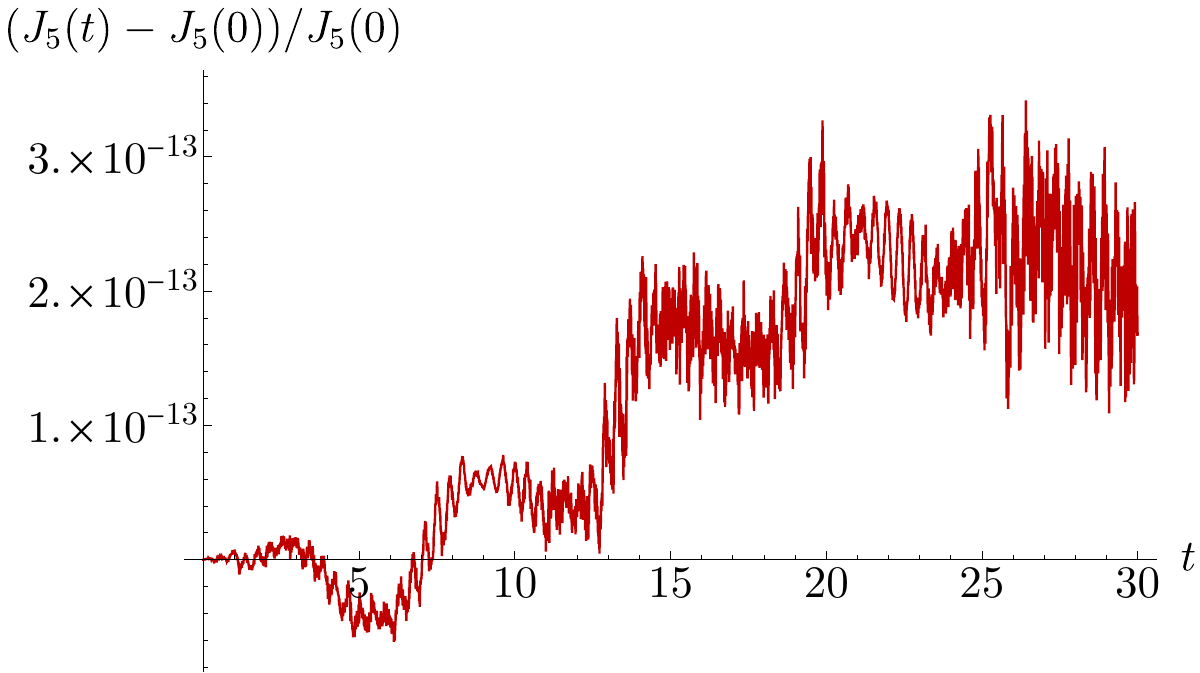}
  }
  \\
  \subfigure[Component $J_{6}$]{
    \includegraphics[width=.31\linewidth]{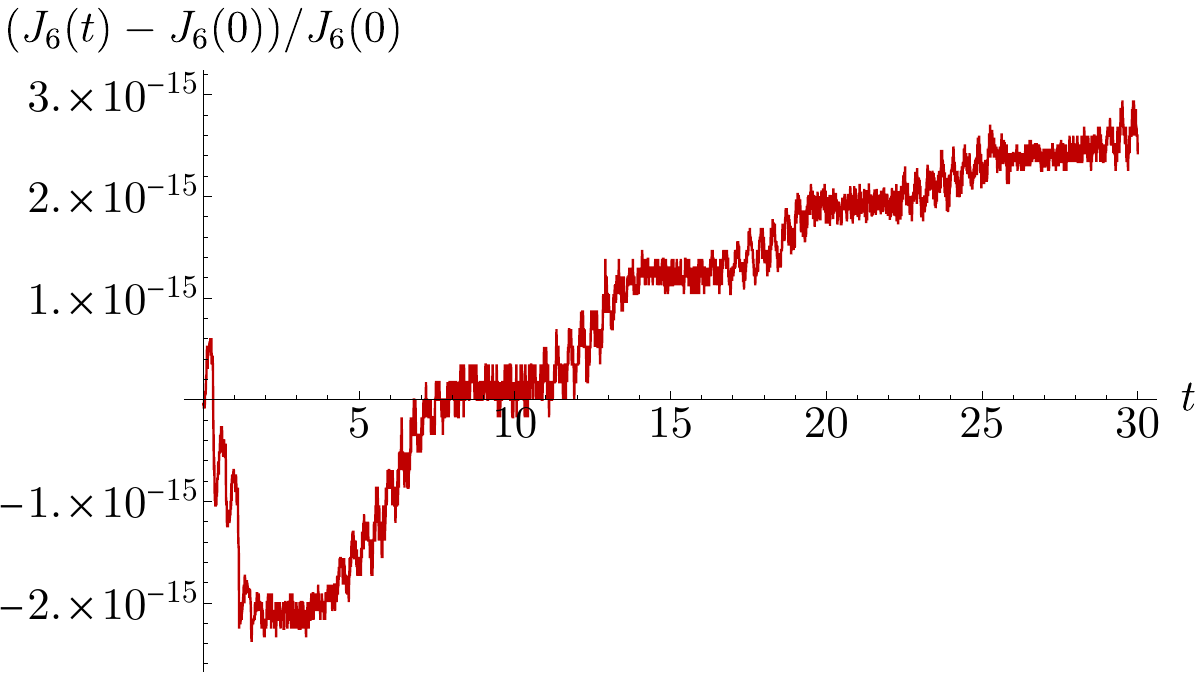}
  }
  \subfigure[Component $J_{7}$]{
    \includegraphics[width=.31\linewidth]{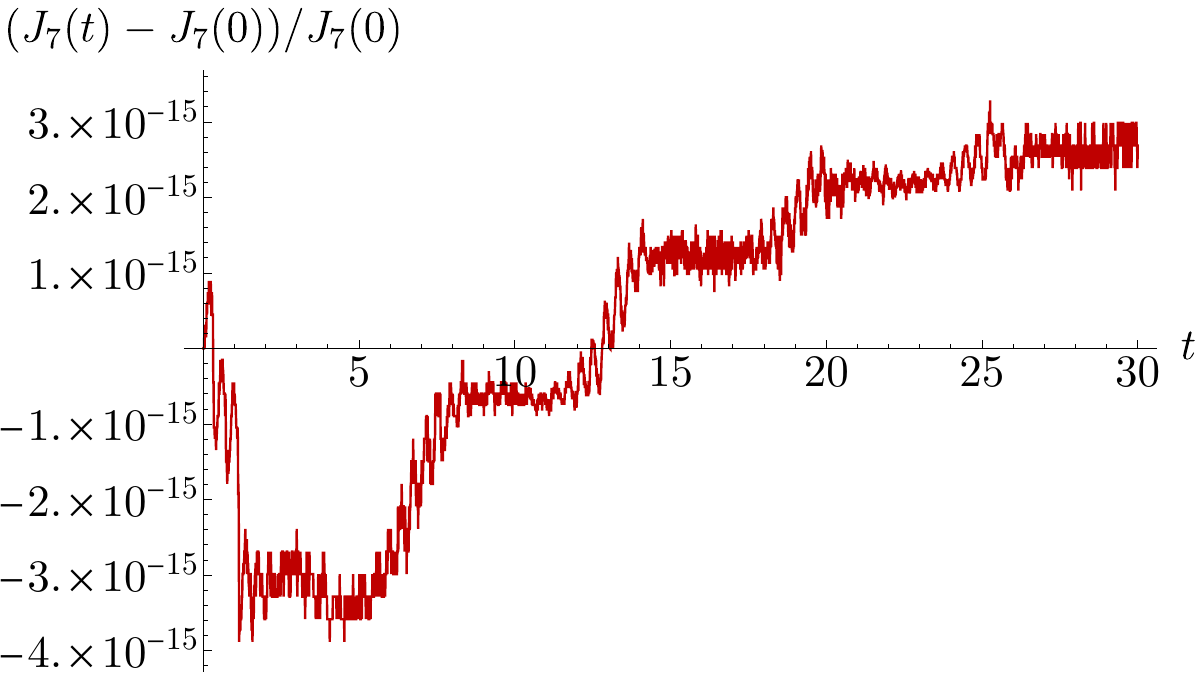}
  }
  \subfigure[Component $J_{8}$]{
    \includegraphics[width=.31\linewidth]{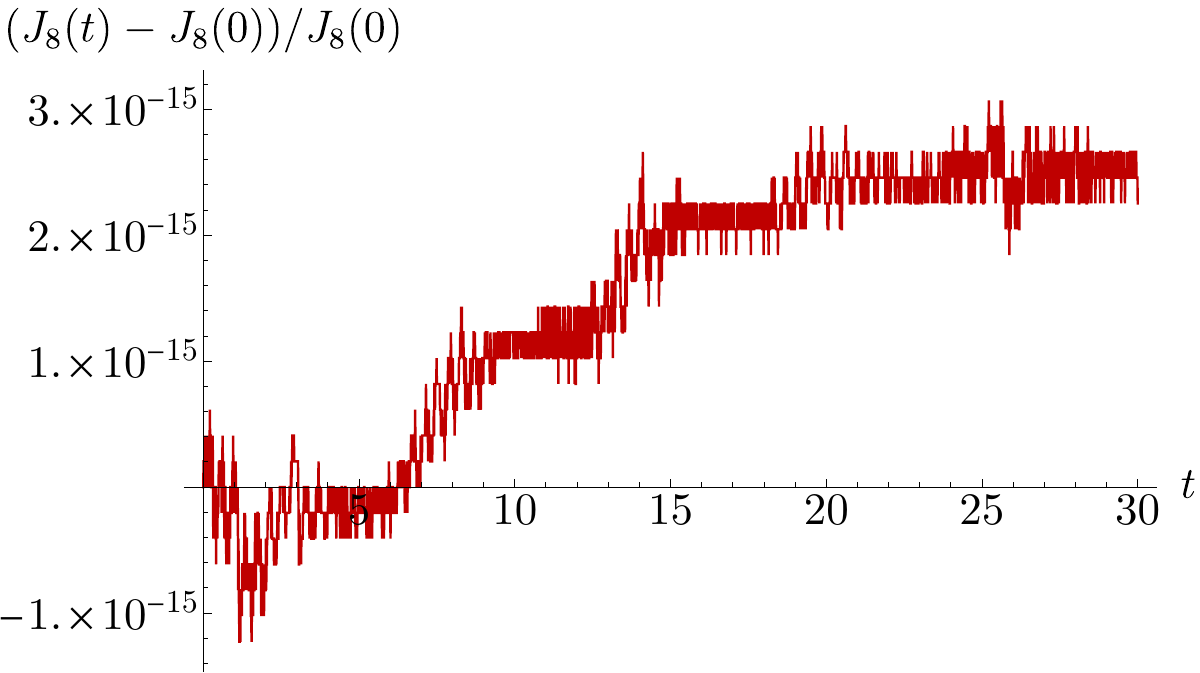}
  }
  \captionsetup{width=0.95\textwidth}
  \caption{
    Time evolutions of errors in components of momentum map $\mathbf{J}$ from \eqref{eq:J-htmb} computed by the $4^{\rm th}$ order Gauss--Legendre method applied to the canonized system for heavy top on a movable base.
    Note that we used the absolute error for $J_{0}$ because $J_{0}(0) = 0$, whereas all the others use relative errors.
    The solutions are shown for the time interval $0 \le t \le 30$ with time step $\Delta t = 0.01$.
  }
  \label{fig:Killing_inv-htmb}
\end{figure}

\section*{Acknowledgments}
We would like to thank the reviewers for their comments and suggestions, particularly the suggestion to find a dual pair.
BJ and TO were partially supported by NSF grants CMMI-1824798 and DMS-2006736.
PJM was supported by the DOE Office of Fusion Energy Sciences under DE-FG02-04ER-54742 and a Forschungspreis from the Alexander von Humboldt Foundation.

\appendix
\numberwithin{equation}{section}
\numberwithin{theorem}{section}
\section{Proof of \Cref{prop:nontriviality_of_h}}
\label{sec:proof-nontriviality_of_h}
Clearly $\sigma_{0}$ satisfies the condition in \eqref{eq:mathfrak_h-characterization}, and so $\sigma_{0} \in \mathfrak{h}$.

Let us next show that $\varkappa \in \mathfrak{h}$.
According to \eqref{eq:mathfrak_h-characterization}, it suffices to show $\mathcal{C}_{i} \kappa = -\kappa\mathcal{C}_{i}^{T}$ for any $i \in \{1, \dots, n\}$.
To that end, first recall that the Jacobi identity for the Lie bracket in $\mathfrak{g}$ is equivalent to the following relationship for the structure constants:
\begin{equation*}
  c^{m}_{il} c^{l}_{jk} + c^{m}_{jl} c^{l}_{ki} + c^{m}_{kl} c^{l}_{ij} = 0.
\end{equation*}
Using this identity, we see that, for any $i, j, l \in \{1, \dots, n\}$,
\begin{align*}
  (\mathcal{C}_{i} \kappa)_{jl}
  = (\mathcal{C}_{i})_{jk} \kappa_{kl}
  &= c^{k}_{ij} c^{r}_{km} c^{m}_{lr} \\
  &= (-c^{r}_{mk} c^{k}_{ij}) c^{m}_{lr} \\
  &= (c^{r}_{ik} c^{k}_{jm} + c^{r}_{jk} c^{k}_{mi}) c^{m}_{lr} \\
  &= c^{r}_{ik} c^{k}_{jm} c^{m}_{lr} + c^{r}_{jk} (-c^{k}_{im} c^{m}_{lr}) \\
  &= c^{m}_{ir} c^{r}_{jk} c^{k}_{lm} + c^{r}_{jk} (c^{k}_{lm} c^{m}_{ri} + c^{k}_{rm} c^{m}_{il}) \\
  &= -(c^{r}_{jk} c^{k}_{mr}) c^{m}_{il} \\
  &= -\kappa_{jm} c^{m}_{il} \\
  &= -\kappa_{jm} (\mathcal{C}_{i})_{lm} \\
  &= -(\kappa \mathcal{C}_{i}^{T})_{jl}.
\end{align*}

Finally, suppose that $\mathfrak{g}$ is semisimple.
Then the Killing form is non-degenerate, i.e., $\kappa$ is invertible, and so $\varkappa^{*}$ is defined.
Now, according to \eqref{eq:mathfrak_h-characterization}, it suffices to show  $\kappa^{-1} \mathcal{C}_{i} = -\mathcal{C}_{i}^{T} \kappa^{-1}$ for any $i \in \{1, \dots, n\}$.
But then this is equivalent to $\mathcal{C}_{i} \kappa = -\kappa\mathcal{C}_{i}^{T}$ that we have shown above.
Hence $\varkappa^{*} \in \mathfrak{h}$ as well.

\section{Proof of \Cref{thm:dual_pair}}
\label{sec:proof-dual_pair}
We prove it only for $\mathbf{M}^{+}$ because the same argument applies to $\mathbf{M}^{-}$ as well.
We break down the proof into a couple of lemmas on the properties of the momentum maps $\mathbf{M}^{+}$ and $\mathbf{J}$.
Note also that, throughout the proof, we identify $T^{*}\mathfrak{g}$ with $T^{*}\R^{n} \cong \R^{2n}$ using the standard bases $\{ E_{i} \}_{i=1}^{n}$ and $\{ E_{*}^{i} \}_{i=1}^{n}$ for $\mathfrak{g}$ and $\mathfrak{g}^{*}$, respectively.

\begin{lemma}
  \label{lem:W-TM}
  Define subspace
  \begin{equation}
    \label{eq:W}
    W(z) \defeq \Span\{ \mathbb{J}\mathcal{M}_{i} z \}_{i=1}^{n} \subset T_{z}T^{*}\mathfrak{g}.
    \quad
    \forall z \in T^{*}\mathfrak{g}.
  \end{equation}
  Then, $\ker T_{z}\mathbf{M}^{+}$ and $W(z)$ are symplectically orthogonal complements to each other, i.e.,
  \begin{equation*}
    \parentheses{ \ker T_{z}\mathbf{M}^{+} }^{\Omega} = W(z)
    \quad
    \forall z \in T^{*}\mathfrak{g}.
  \end{equation*}
\end{lemma}
\begin{proof}
  Let us first show that $W(z)$ and $\ker T_{z}\mathbf{M}^{+}$ are complementary in dimensions.
  To that end, let us write the tangent map of $\mathbf{M}^{+}\colon T^{*}\mathfrak{g} \to \mathfrak{g}^{*}$ using the components~\eqref{eq:M-components} for $\mathbf{M}^{+}$:
  \begin{equation}
    \label{eq:TM-A}
    T_{z}\mathbf{M}^{+} (\dot{z}) =
    \begin{bmatrix}
      \ip{ \d{M_{1}}(z) }{ \dot{z} } \\
      \vdots \\
      \ip{ \d{M_{n}}(z) }{ \dot{z} } \\
    \end{bmatrix} =
    \begin{bmatrix}
      z^{T}\mathcal{M}_{1} \dot{z} \\
      \vdots \\
      z^{T}\mathcal{M}_{n} \dot{z}
    \end{bmatrix}
    = A(z) \dot{z}
    \quad
    \text{with}
    \quad
    A(z) \defeq
    \begin{bmatrix}
      z^{T}\mathcal{M}_{1} \\
      \vdots \\
      z^{T}\mathcal{M}_{n}
    \end{bmatrix}
    \in \R^{n \times 2n}
  \end{equation}
  for any $z = (q,p) \in T\mathfrak{g}^{*}$ and any $\dot{z} = (\dot{q}, \dot{p}) \in T_{z}T^{*}\mathfrak{g}$.
  Therefore, we find
  \begin{equation*}
    \ker T_{z}\mathbf{M}^{+} = \ker A(z).
  \end{equation*}
  Now, by the fundamental theorem of linear algebra, we see that
  \begin{equation*}
    \im\parentheses{ A(z)^{T} } = \Span\{ \mathcal{M}_{i} z \}_{i=1}^{n} \subset T_{z}T^{*}\mathfrak{g} \cong \R^{2n}
  \end{equation*}
  gives a complementary subspace to $\ker T_{z}\mathbf{M}^{+}$ in $T_{z}T^{*}\mathfrak{g} \cong \R^{2n}$, i.e.,
  \begin{equation*}
    \dim \parentheses{ \im\parentheses{ A(z)^{T} } } + \dim\!\parentheses{ \ker T_{z}\mathbf{M}^{+} } = 2n.
  \end{equation*}
  But then, since $\mathbb{J}$ is non-degenerate, we see that $\dim W(z) = \dim \parentheses{ \im\parentheses{ A(z)^{T} } }$.
  Therefore,
  \begin{equation*}
    \dim W(z) + \dim\!\parentheses{ \ker T_{z}\mathbf{M}^{+} } = 2n
    \quad \forall z \in T^{*}\mathfrak{g}.
  \end{equation*}

  It remains to show that $W(z)$ is symplectically orthogonal to $\ker T_{z}\mathbf{M}^{+}$.
  Let $\dot{z} \in \ker T_{z}\mathbf{M}^{+} = \ker A(z)$ be arbitrary.
  Then $z^{T} \mathcal{M}_{i} \dot{z} = 0$ for any $i \in \{1, \dots, n\}$, but then this implies
  \begin{align*}
    \Omega\parentheses{ \mathbb{J}\mathcal{M}_{i} z, \dot{z} }
    &= (\mathbb{J}\mathcal{M}_{i} z)^{T} \mathbb{J} \dot{z} \\
    &= z^{T} \mathcal{M}_{i} \dot{z} \\
    &= 0. \qedhere
  \end{align*}
\end{proof}

Let us next prove some properties of the other momentum map $\mathbf{J}$ in the pair.
\begin{lemma}
  \label{lem:J}
  The momentum map $\mathbf{J}\colon T^{*}\mathfrak{g} \to \mathfrak{h}^{*}$ satisfies the following for any $z \in T^{*}\mathfrak{g}$:
  \begin{enumerate}
  \item $J_{\sigma}(z) \defeq \ip{ \mathbf{J}(z) }{ \sigma } = \frac{1}{2}z^{T} \sigma z$ for any $\sigma \in \mathfrak{h}$; \smallskip
    \label{part:J}
  \item $W(z) \subset \ker T_{z}\mathbf{J}$,
    \label{part:WinTJ}
  \end{enumerate}
  where $W(z)$ is the subspace of $T_{z}T^{*}\mathfrak{g}$ defined in \eqref{eq:W}.
\end{lemma}
\begin{proof}
  Let us first show \eqref{part:J}.
  By the definition of momentum map, we seek $J_{\sigma}(\,\cdot\,) \defeq \ip{ \mathbf{J}(\,\cdot\,) }{ \sigma }\colon T^{*}\mathfrak{g} \to \R$ satisfying, for any $\sigma \in \mathfrak{h}$ and any $z \in T^{*}\mathfrak{g}$,
  \begin{equation*}
    \sigma_{T^{*}\mathfrak{g}}(z) = X_{J_{\sigma}}(z),
  \end{equation*}
  where $X_{J_{\sigma}}(z)$ is the Hamiltonian vector field defined by $J_{\sigma}$, i.e.,
  \begin{equation*}
    X_{J_{\sigma}}(z) = \mathbb{J} \nabla J_{\sigma}(z).
  \end{equation*}
  Therefore, we have $\sigma z = \nabla J_{\sigma}(z)$, and thus $J_{\sigma}(z) =  \frac{1}{2}z^{T} \sigma z$.
  
  For \eqref{part:WinTJ}, it suffices to show that $W(z) \subset \ker \d{J_{\sigma}}(z)$ for any $z \in T^{*}\mathfrak{g}$ and any $\sigma \in \mathfrak{h}$.
  First recall from \eqref{eq:mathfrak_h} that $\sigma \in \mathfrak{h}$ if and only if $\sigma \mathbb{J} \mathcal{M}_{i}$ is skew-symmetric for any $i \in \{1, \dots, n\}$.
  Therefore, we see that, for any $i \in \{1, \dots, n\}$,
  \begin{equation*}
    \ip{ \d{J_{\sigma}}(z) }{ \mathbb{J} \mathcal{M}_{i} z } = z^{T} \sigma \mathbb{J} \mathcal{M}_{i} z = 0.
  \end{equation*}
  Since $W(z)$ is spanned by $\{ \mathbb{J}\mathcal{M}_{i} z \}_{i=1}^{n}$, we have $W(z) \subset \ker \d{J_{\sigma}}(z)$.
\end{proof}

We are now ready to prove \Cref{thm:dual_pair}.

\begin{enumerate}[(i)]
\item This is \Cref{lem:J}~\eqref{part:J}.
\item In order to find the open subset $U \subset T^{*}\mathfrak{g}$, notice first that $T_{z}\mathbf{M}^{+} = A(z)$ (see \eqref{eq:TM-A}) is full-rank if and only if $A(z) A(z)^{T}$ is non-singular.
  However,
  \begin{equation*}
    A(z) A(z)^{T}
    = \begin{bmatrix}
      z^{T}\mathcal{M}_{1} \\
      \vdots \\
      z^{T}\mathcal{M}_{n}
    \end{bmatrix}
    \begin{bmatrix}
      \mathcal{M}_{1} z \dots \mathcal{M}_{n} z
    \end{bmatrix}
    = \begin{bmatrix}
      \norm{ \mathcal{M}_{1} z }^{2} & \dots & z^{T} \mathcal{M}_{1} \mathcal{M}_{n} z \\
      \vdots & \ddots & \vdots \\
      z^{T} \mathcal{M}_{n} \mathcal{M}_{1} z  & \dots & \norm{ \mathcal{M}_{n} z }^{2}
    \end{bmatrix}.
  \end{equation*}
  Since each entry is quadratic in $z$, the function $d_{A}(z) \defeq \det\parentheses{ A(z) A(z)^{T} }$ is a polynomial of $z$ as well.
  Hence the pre-image $U_{1} \defeq d_{A}^{-1}(\R\backslash\{0\}) \subset T^{*}\mathfrak{g}$ is an open set on which $T\mathbf{M}^{+}$ is full-rank.
  Let $\{ \sigma_{j} \}_{j=0}^{m-1}$ be a basis for $\mathfrak{h}$, and define
  \begin{equation*}
    B(z) \defeq
    \begin{bmatrix}
      z^{T}\sigma_{0} \\
      \vdots \\
      z^{T}\sigma_{m-1}
    \end{bmatrix}.
  \end{equation*}
  Then, running the same argument with $B(z)$ in place of $A(z)$, one can find $U_{2} \defeq d_{B}^{-1}(\R\backslash\{0\})$ on which $T\mathbf{J}$ is full-rank.
  Then $U \defeq U_{1} \cap U_{2}$ gives the desired open set.
\item Note first that $\mathbf{M}^{+}$ and $\mathbf{J}$ being submersions imply that they are submersions onto open sets $\mathbf{M}^{+}(U) \subset \mathfrak{g}^{*}$ and $\mathbf{J}(U) \subset \mathfrak{h}^{*}$, respectively.
  Therefore, we have, for any $z \in U$,
  \begin{equation*}
    \dim\!\parentheses{ \ker T_{z}\mathbf{M}^{+} }
    = 2n - \dim\!\parentheses{ \im T_{z}\mathbf{M}^{+} }
    = 2n - \dim\mathfrak{g},
  \end{equation*}
  and
  \begin{align*}
    \dim\!\parentheses{ \ker T_{z}\mathbf{J} }   
    = 2n - \dim\!\parentheses{ \im T_{z}\mathbf{J} }
    = 2n - \dim\mathfrak{h}.
  \end{align*}
  Thus the assumption $\dim\mathfrak{h} = \dim\mathfrak{g} = n$ implies that $\dim\!\parentheses{ \ker T_{z}\mathbf{M}^{+} } = \dim\!\parentheses{ \ker T_{z}\mathbf{J} } = n$ for any $z \in U$.
  However, using \Cref{lem:W-TM} and \Cref{lem:J}~\eqref{part:WinTJ}, we have
  \begin{equation*}
    \parentheses{ \ker T_{z}\mathbf{M}^{+} }^{\Omega} \subset \ker T_{z}\mathbf{J}
    \quad
    \forall z \in T^{*}\mathfrak{g}.
  \end{equation*}
  But then $\dim \parentheses{ \ker T_{z}\mathbf{M}^{+} }^{\Omega}= n = \dim\!\parentheses{ \ker T_{z}\mathbf{J} }$ for any $z \in U$, and thus we obtain
  \begin{equation*}
    \parentheses{ \ker T_{z}\mathbf{M}^{+} }^{\Omega} = \ker T_{z}\mathbf{J}
    \quad
    \forall z \in U.
  \end{equation*}
\end{enumerate}

\bibliography{Clebsch}
\bibliographystyle{plainnat}

\end{document}